\documentclass[journal]{IEEEtran}
\IEEEoverridecommandlockouts
\usepackage{cite}
\usepackage{ amsmath,amssymb,amsfonts, amsthm, enumitem}
\usepackage{algorithm}
\usepackage{algpseudocode}
\usepackage{graphicx}
\usepackage{textcomp}
\usepackage{xcolor}
\usepackage{array}
\usepackage{siunitx}
\usepackage{upgreek}
\usepackage{commath}
\usepackage{url}
\usepackage{xparse}
\usepackage{mathtools}
\usepackage{float}
\usepackage{tabularx}
\usepackage{xcolor}
\usepackage[hang,flushmargin]{footmisc} 
\usepackage{wrapfig}
\usepackage{amsthm}
\usepackage[]{footmisc}
\usepackage{caption}
\expandafter\def\expandafter\UrlBreaks\expandafter{\UrlBreaks
\do \d}

\newcolumntype{P}[1]{>{\centering\arraybackslash}p{#1}}
\newcolumntype{M}[1]{>{\centering\arraybackslash}m{#1}}

\let\svthefootnote\thefootnote

\NewDocumentCommand{\INTERVALINNARDS}{ m m }{
    #1 {,} #2
}
\NewDocumentCommand{\interval}{ s m >{\SplitArgument{1}{,}}m m o }{
    \IfBooleanTF{#1}{
        \left#2 \INTERVALINNARDS #3 \right#4
    }{
        \IfValueTF{#5}{
            #5{#2} \INTERVALINNARDS #3 #5{#4}
        }{
            #2 \INTERVALINNARDS #3 #4
        }
    }
}

\mathchardef\Re="023C
\mathchardef\Im="023D

\def\BibTeX{{\rm B\kern-.05em{\sc i\kern-.025em b}\kern-.08em
    T\kern-.1667em\lower.7ex\hbox{E}\kern-.125emX}}
\usepackage{caption, subcaption}
\newcommand{\defeq}{\vcentcolon=}
\newcommand\aeq{\stackrel{\mathclap{\normalfont\mbox{(a)}}}{=}}

\newcommand\beq{\stackrel{\mathclap{\normalfont\mbox{(b)}}}{=}}
\newcommand\bg{\stackrel{\mathclap{\normalfont\mbox{(b)}}}{>}}
\newcommand\bl{\stackrel{\mathclap{\normalfont\mbox{(b)}}}{<}}
\newcommand\ceq{\stackrel{\mathclap{\normalfont\mbox{(c)}}}{=}}
\newcommand\cl{\stackrel{\mathclap{\normalfont\mbox{(c)}}}{<}}
\newcommand\cg{\stackrel{\mathclap{\normalfont\mbox{(c)}}}{>}}
\newcommand\deq{\stackrel{\mathclap{\normalfont\mbox{(d)}}}{=}}
\newcommand\dl{\stackrel{\mathclap{\normalfont\mbox{(d)}}}{<}}
\newcommand\dg{\stackrel{\mathclap{\normalfont\mbox{(d)}}}{>}}
\newcommand\eeq{\stackrel{\mathclap{\normalfont\mbox{(e)}}}{=}}
\newcommand\el{\stackrel{\mathclap{\normalfont\mbox{(e)}}}{<}}
\newcommand\eg{\stackrel{\mathclap{\normalfont\mbox{(e)}}}{>}}
\newcommand\feq{\stackrel{\mathclap{\normalfont\mbox{(f)}}}{=}}
\newcommand\fl{\stackrel{\mathclap{\normalfont\mbox{(f)}}}{<}}

\newcommand\ggl{\stackrel{\mathclap{\normalfont\mbox{(g)}}}{<}}

\newcommand\hheq{\stackrel{\mathclap{\normalfont\mbox{(h)}}}{=}}
\newcommand\hl{\stackrel{\mathclap{\normalfont\mbox{(h)}}}{<}}

\newcommand\il{\stackrel{\mathclap{\normalfont\mbox{(i)}}}{<}}

\newtheorem{prop}{Proposition}

\begin{document}
\setlength{\abovedisplayskip}{5pt}
\setlength{\belowdisplayskip}{5pt}

\font\myfont=cmr12 at 15pt
\title{{  Reconstructing Classes of Non-bandlimited Signals from Time Encoded Information\\
}}
\author{\IEEEauthorblockN{Roxana Alexandru, \textit{Student Member}, IEEE, and Pier Luigi Dragotti, \textit{Fellow}, IEEE}
\vspace{-1.5em}
}

\vspace{-2em}
\maketitle

\vspace{-2em}
\begin{abstract}
We investigate time encoding as an alternative method to classical sampling, and address the problem of reconstructing classes of non-bandlimited signals from time-based samples.
We consider a sampling mechanism based on first filtering the input, before obtaining the timing information using a time encoding machine. Within this framework, we show that sampling by timing is equivalent to a non-uniform sampling problem, where the reconstruction of the input depends on the characteristics of the filter and on its non-uniform shifts. The classes of filters we focus on are exponential and polynomial splines, and we show that their fundamental properties are locally preserved in the context of non-uniform sampling.
Leveraging these properties, we then derive sufficient conditions and propose novel algorithms for perfect reconstruction of classes of non-bandlimited signals such as: streams of Diracs, sequences of pulses and piecewise constant signals. Next, we extend these methods to operate with arbitrary filters, and also present simulation results on synthetic noisy data.
\end{abstract}
\vspace{-0.5em}

\IEEEpeerreviewmaketitle
\vspace{-0.5em}
\section{Introduction}
Sampling plays a fundamental role in signal processing and communications, achieving the conversion of continuous time phenomena into discrete sequences \cite{843002}. 
From the Whittaker-Shannon theorem \cite{1697831}, to recent theories in compressed sensing \cite{1580791,1614066}, super-resolution \cite{candes-granda} and finite rate of innovation \cite{4156380,1003065,5686950,4483755, 4682542}, sampling theory has provided precise answers on when a faithful conversion of a continuous waveform into a discrete sequence is possible. These methods are generally based on recording the signal amplitude at specified times, which lead to uniform sampling if the samples are evenly spaced, and non-uniform sampling otherwise.

In this paper, we concentrate on an alternative method to classical sampling, which encodes the input into a sequence of non-uniformly spaced time events or \textit{spikes}. In other words, rather than recording the value of the signal at preset times, one records the instants when the signal crosses a pre-defined threshold or triggers a pre-defined event. 
\let\thefootnote\relax\footnote{Some of the work in this paper was, in part, presented at the IEEE International Conference on Acoustics, Speech and Signal Processing (ICASSP), Brighton, UK, May 2019 \cite{8682626}, and the International Conference on Sampling Theory and Applications (SampTA), Bordeaux, France, July 2019 \cite{alexandrusampta}.}Acquisition models inspired by this mechanism include zero-crossing detectors \cite{6770840}, delta-modulation schemes\cite{soton252088}, as well as the time encoding machine (TEM) introduced in \cite{1344228}.
This latter model is of particular interest, as it mimics the \textit{integrate-and-fire} mechanism of neurons in the human brain. Biological neurons use time encoding to represent sensory information as action potentials \cite{adrian1928basis, Dayan:2005:TNC:1205781, Gerstner:2002:SNM:583784}, which allows them to process information very efficiently. In the same manner, sampling inspired by the brain could lead to very simple and highly efficient devices, ranging from analog to digital converters \cite{1344228}, to neuromorphic computing or event-based vision sensors, which record only changes in the input intensity, leading to low power consumption and fewer storage requirements \cite{5537149}. 

At the same time, time-encoding methods extend theories of traditional sampling, and this makes this topic intriguing also from a research perspective. 
Within the study of time encoding, the key problem that arises is to find methods to retrieve the input signal from its timing information, and hence the key questions to pursue are the following. 1) Is time encoding invertible, and which classes of signals can be uniquely represented using timing information? 2) What algorithms allow perfect retrieval of these signals from their time-encoded samples? 

To address these questions, several authors have provided ways to sample and reconstruct bandlimited signals \cite{Lazar05timeencoding, 1201780, 5709990, Feichtinger2012, 1415989, Adam19}. These initial results on time-encoding machines have also been extended to functions that belong to shift-invariant spaces \cite{ FlorescuC15 , GONTIER201463}, typically by connecting time encoding with the problem of non-uniform sampling \cite{AldroubiGrochenig, article,330352}.
Time encoding theory has also been generalized to the case of non-bandlimited signals in \cite{LAP09b}, however in the context of studying the dynamics of populations of neurons, by leveraging stochastic assumptions on the firing parameters.

In this paper, we show that it is possible to perfectly reconstruct particular classes of continuous-time signals which are neither bandlimited nor belong to shift-invariant subspaces, from samples obtained using a time encoding mechanism.
The signals we focus on are infinite streams of Diracs, sequences of pulses, as well as piecewise constant signals.
Sampling and reconstructing pulses is of significant relevance to many real-world applications. For example, \textit{time-of-flight} cameras probe the 3D scene with pulses of light and reconstruct the scene by measuring their round trip time. In applications which require reduced computational power and speed, e.g. robots mapping their surroundings, time-of-flight technology may benefit from a time encoding framework which would significantly lower the sampling rate. Signals consisting of a stream of pulses appear in many other applications, including: ultrawideband communications \cite{1329542}, ECG acquisition and compression \cite{7857059}, radio-astronomy \cite{7736135}, image processing\cite{7465789}, ultrasound imaging \cite{5686950} and processing of neuronal signals \cite{Onativia2013}.

At the same time, time encoding principles have already been integrated in bio-inspired technologies such as \textit{dynamic vision sensors} (DVS) \cite{5537149}, which have many real-world applications, ranging from robotics to autonomous driving as well as low-power surveillance.
In a DVS camera, each pixel only records changes in the input at the time instants they occur, by taking a time derivative of the signal. Hence, at the local pixel-level, this is equivalent to time encoding of piecewise constant signals, which is studied in this paper. 

Motivated by these real-world applications, the time encoding strategy we propose is based on filtering the input signal before extracting the timing information using a crossing or an integrate-and-fire TEM. The filter may be used to reduce noise, or may model the distortion introduced by the acquisition device, for example the optics in a time-of-flight scanner or the photoreceptors in a DVS camera. 
In order to develop a framework for exact reconstruction, we initially focus on two classes of compact-support filters (sampling kernels): exponential and polynomials splines. 
Please note that exponential splines are very useful since they can be used to model any convolution operator with rational transfer function as for example, simple RC circuits \cite{1408194, 4156380}.
Our first main contribution is to prove that exponential (polynomial) splines locally preserve their exponential (polynomial) reproducing properties in the context of time-based sampling. Specifically, we show that within intervals where there are no knots of at least $N$ non-uniformly shifted kernels, we can locally reproduce exponentials (polynomials) of degree $N$. The second aspect of our contribution is to leverage these properties to address the problem of reconstructing some classes of non-bandlimited signals from timing information. We initially develop our reconstruction framework for the case of one Dirac, where we show how a linear combination of its non-uniform samples leads to a sequence of signal moments, which can then be annihilated using Prony's method \cite{Prony}, in order to retrieve the free parameters of the input. Furthermore, we extend this method to reconstruct infinite streams and bursts of Diracs, sequences of pulses as well as piecewise constant signals, for which we can achieve \textit{local} reconstruction given the compact support of the filter. Finally, we depart from the ideal case, and present a universal reconstruction strategy that works with timing-based samples taken by arbitrary kernels.
 
This paper is organized as follows. In Section \ref{subsection:Acquisition Models}, we describe the principles of time encoding, with two exemplary cases. Then, in Section \ref{subsection:Sampling Kernels} we show that sampling kernels which reproduce exponentials or polynomials preserve this property locally, when sampling is based on timing information. Furthermore,  in Section \ref{section:Perfect Recovery of Signals from Timing Information obtained with a Crossing Time Encoding Machine} we present methods for the reconstruction of non-bandlimited signals from their timing information obtained using a crossing TEM. We first propose a method for estimation of a single Dirac, and extend this to retrieve streams of Diracs and bursts of Diracs. Then, in Section \ref{section:Perfect Recovery of Signals from Timing Information obtained with an Integrate-and-fire System} we demonstrate the perfect retrieval of classes of non-bandlimited signals from timing information, obtained using an integrate-and-fire TEM. These estimation methods are then extended in Section \ref{section:Generalized Time-based Sampling} to the case of arbitrary sampling kernels. Here we also present results for the case of noisy signals. Finally, we highlight the high efficiency of sampling based on timing information in Section \ref{section:Density of Non-uniform Samples}, and present concluding remarks in Section \ref{sec:Conclusions}. Please note that the code to reproduce our simulations is available online \cite{coderoxana}.

\vspace{-1em}
\section{Time Encoding Mechanisms}
\label{sec:Time-based Sampling Mechanism}

\subsection{Acquisition Models}
\label{subsection:Acquisition Models}
In this section, we introduce the time encoding machines considered in this paper: the crossing TEM and the integrate-and-fire TEM. Specifically, we show how these TEMs map a real signal $x(t)$ to a strictly increasing sequence of times $\{t_n\}$ \cite{GONTIER201463}. We also show that although no measure of the amplitude of the signal is recorded, time encoding is equivalent to a non-uniform sampling problem.

\vspace{0.2em}
\subsubsection{Crossing Time Encoding Machine}
The crossing time encoding strategy is inspired by the A/D conversion scheme in e.g. \cite{838174,  GONTIER201463}, and is depicted in Fig.~\ref{fig:comparator}. It consists of a compact-support filter $\varphi(-t)$, and a comparator with a sinusoidal reference $g(t)$.
The output of the acquisition device is the sequence $\{t_n\}$, corresponding to the time instants when the filtered input signal crosses the reference, i.e. when $y(t_n)-g(t_n)=0$. Moreover, since the shape of the test function $g(t)$ is known, we can retrieve the amplitudes of the output samples, given by $y_n=y(t_n)=g(t_n)$. Hence, decoding the input signal is equivalent to a non-uniform sampling problem, where we aim to reconstruct $x(t)$ from the non-uniform samples given by:
\small
\begin{equation}\label{eq:non-uniform samples comparator}
y_n = y(t_n)= \int x(\tau)\varphi(\tau-t_n)d\tau=\langle x(t), \varphi(t-t_n)\rangle.
\end{equation}
\normalsize

In Fig.~\ref{fig:comparator_info} we depict the time encoded information of an input signal of 3 Diracs, obtained using the TEM in Fig.~\ref{fig:comparator}.

\small
\begin{figure}[htb]
\vspace{-0.5em}
\centering
\includegraphics[width=0.32\textwidth]{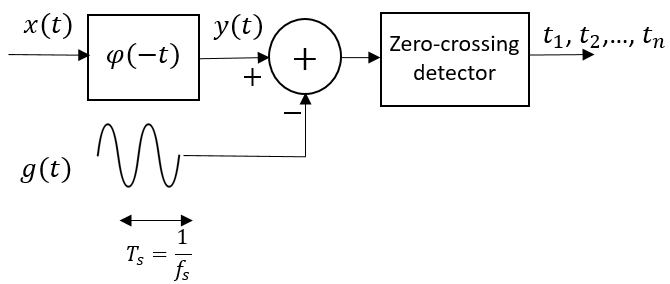}
\caption{Crossing Time Encoding Machine.}
\vspace{-0.5em}
\label{fig:comparator}
\vspace{-0.5em}
\end{figure}
\normalsize

\begin{figure}[htb]
\vspace{-0.5em}
\centering
\includegraphics[width=0.5\textwidth]{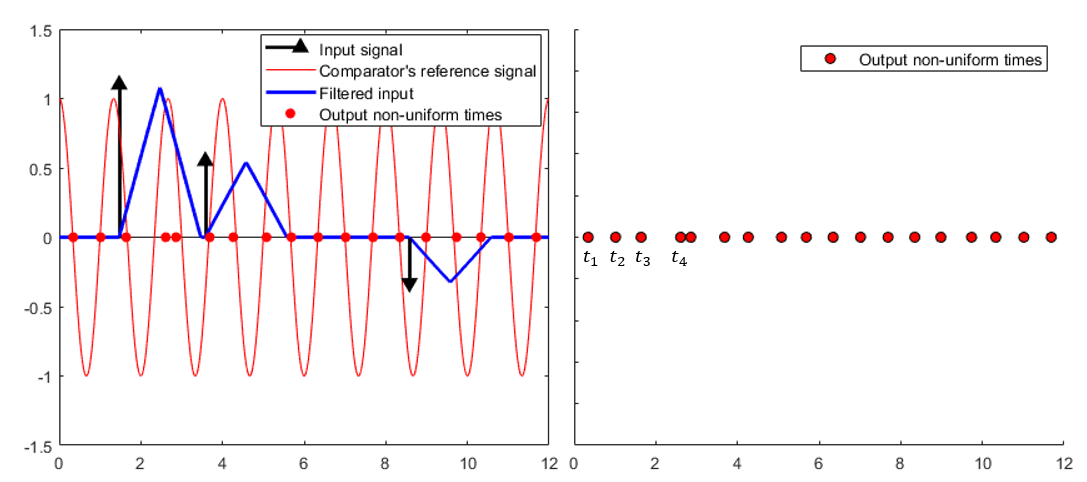}
\caption{Time encoding based on the Crossing TEM.}
\label{fig:comparator_info}
\end{figure}

\subsubsection{Time Encoding based on an Integrate-and-fire System}
The operating principle of this time encoding strategy is similar to the one in \cite{Lazar05timeencoding}, and is depicted in Fig.~\ref{fig:integrate_and_fire_model_2}.
The signal is first filtered with a compact-support filter with impulse response $\varphi(-t)$, before being passed to an integrator. When the output of the integrator reaches the positive trigger mark $C_T$, the time encoding machine outputs a spike and the integrated signal $y(t)$ is reset to zero. Similarly, a spike is generated and $y(t)$ resets to zero, when the integrator reaches the negative trigger mark $-C_T$.
The time instants when the integrator reaches the threshold $\pm C_T$ are recorded in the sequence $\{t_n\}$. 
Then, we can compute the output sample $y(t_n)$ at each spike $t_n$ as:
\small
\begin{equation} 
\label{eq:non_uniform_samples integrator}
y_n=y(t_n)=\pm C_T=\int_{t_{n-1}}^{t_n} f(\tau) d\tau,
\end{equation}
\normalsize
where $n \geq 2$ and $f(t)$ is defined as:

\small
\begin{equation}
\label{eq:filtered input integrator}
f(t)=\int  x(\alpha) \varphi(\alpha-t) d\alpha, \text{ for } t \in [t_{n-1}, t_n].
\end{equation}
\normalsize

Similarly, assuming that the input signal $x(t)=0$, for $t<\tau_1$, and that the filter $\varphi(-t)$ is causal, then the first output sample is given by:
\small
\begin{equation} 
\label{eq:non_uniform_sample_1 integrator}
y_1=y(t_1)=\pm C_T= \int_{\tau_1}^{t_1} f(\tau) d\tau.
\end{equation}
\normalsize

Hence, time encoding with an integrate-and-fire model is equivalent to a non-uniform sampling problem, where we aim to estimate the input $x(t)$ from the non-uniform samples $y(t_n)$.
In Fig.~\ref{fig:int_info} we depict the time encoding of an input signal, obtained using the device in Fig.~\ref{fig:integrate_and_fire_model_2}, for $C_T=0.15$.

\begin{figure}[htb]
\vspace{-0.5em}
\centering
\includegraphics[width=0.5\textwidth]{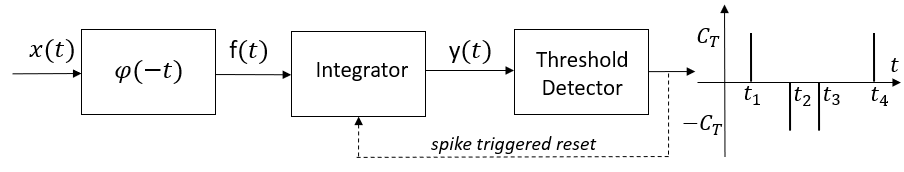}
\caption{Time Encoding Machine based on Integrate-and-fire.}
\label{fig:integrate_and_fire_model_2}
\vspace{-0.5em}
\end{figure}
\vspace{-0.8em}
\begin{figure}[htb]
\vspace{-0.5em}
\centering
\includegraphics[width=0.5\textwidth]{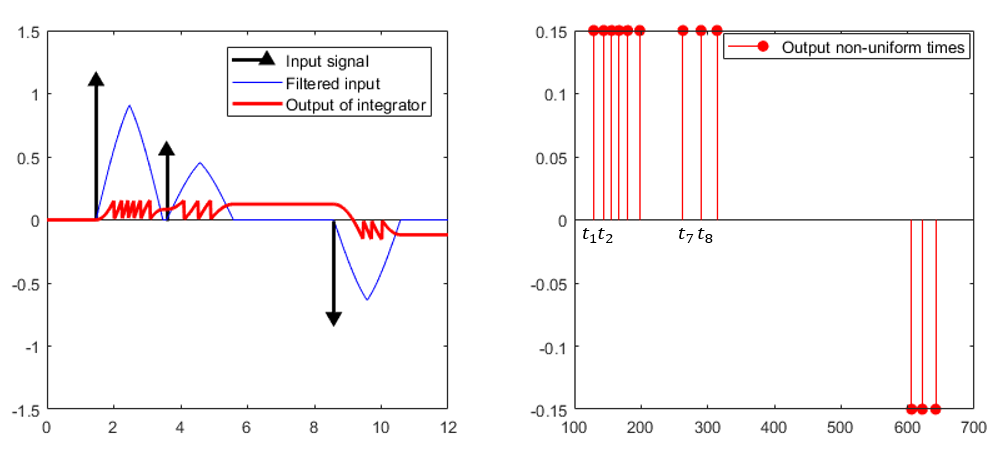}
\caption{Time Encoding based on the Integrate-and-fire TEM.}
\label{fig:int_info}
\end{figure}

Furthermore, leveraging the results in \cite{Lazar2008}, we can show that the non-uniform output samples we obtain using the acquisition model in Fig.~\ref{fig:integrate_and_fire_model_2} are the same as those obtained by filtering the input with the modified kernel $(\varphi*q_{\theta_n})(t)$:
\small
\begin{equation}
\label{eq:modified samples integrator summary}
y(t_n)=\langle x(t), (\varphi*q_{\theta_n})(t-t_{n-1}) \rangle,
\end{equation}
\normalsize
where $\theta_n=t_n-t_{n-1}$ and $q_{\theta_n}(t)$ is defined as:
\small
\begin{equation}
\label{eq:box function integrator}
q_{\theta_n}(t)=
\begin{cases} 1, & 0\leq t\leq \theta_n, \\ 0, & otherwise.
\end{cases}
\end{equation}
\normalsize

We can prove Eq. (\ref{eq:modified samples integrator summary}) by re-writing Eq. (\ref{eq:non_uniform_samples integrator}) as follows:
\small
\begin{equation}
\label{eq: output_samples_equivalent_filter}
\begin{split}
y(t_n) &= \int_{t_{n-1}}^{t_n} f(\tau) d\tau =\int_{t_{n-1}}^{t_n} \int_{-\infty}^{\infty} x(t) \varphi(t-\tau) dt d\tau \\
&\aeq \int_{-\infty}^{\infty} x(t) \int_{t_{n-1}}^{t_n} \varphi(t-\tau) d\tau dt \\
&\beq \int_{-\infty}^{\infty} x(t) \int_{t-t_n}^{t-t_{n-1}}\varphi(\tau) d\tau dt \\
&\ceq \int_{-\infty}^{\infty} x(t) \int_{t-t_n}^{t-t_{n-1}}\varphi(\tau)q_{\theta_n}(t-t_{n-1}-\tau)d\tau dt \\
&\deq\int_{-\infty}^{\infty} x(t) (\varphi*q_{\theta_n})(t-t_{n-1}) dt \\
&=\langle x(t), (\varphi*q_{\theta_n})(t-t_{n-1}) \rangle.
\end{split}
\end{equation}
\normalsize

In the derivations above, $(a)$ holds since we assume both the input $x(t)$ and the filter $\varphi(t)$ have compact support, and $(b)$ follows from a change of variable. Moreover, $(c)$ follows from the fact that $q_{\theta_n}(t-t_{n-1}-\tau)=1$ for $\tau \in [t-t_n, t-t_{n-1}]$ and $(d)$ holds since $q_{\theta_n}(t-t_{n-1}-\tau)=0$ for $\tau \notin [t-t_n, t-t_{n-1}]$, as defined in Eq. (\ref{eq:box function integrator}).

Finally, the first output sample can be computed as:
\small
\begin{equation}
\label{eq: output_sample_1_equivalent_filter}
y(t_1)\aeq\int_{\tau_1}^{t_1} f(\tau) d\tau=\langle x(t), (\varphi*q_{\theta_1})(t-\tau_1) \rangle,
\end{equation}
\normalsize
where $\theta_1=t_1-\tau_1$, and $(a)$ follows from Eq. (\ref{eq:non_uniform_sample_1 integrator}).

We conclude this subsection by making the following remark.
We observe that from the timing sequence $\{t_n\}$, we can either recover $y(t_n)=\langle x(t), \varphi(t-t_{n})\rangle$ for the case of the crossing TEM or $y(t_n)=\langle x(t), (\varphi*q_{\theta_n})(t-t_{n-1})\rangle$ for the integrate-and-fire model. This means that in both cases, the reconstruction of $x(t)$ will depend on the proper choice of the sampling kernel $\varphi(t)$ and on its non-uniform shifts $\varphi(t-t_n)$.

In what follows we focus on two families of kernels, polynomial and exponential splines \cite{4156380, 799930, 1408193}, and show that some of their fundamental properties are preserved in the case of non-uniform shifts.

\vspace{-0.8em}
\subsection{Sampling Kernels}
\label{subsection:Sampling Kernels}
The sampling kernels $\varphi(t)$, that we consider in this paper are all anti-causal since they are the time reversed versions of causal filters.
\subsubsection{Polynomial splines}
\label{subsubsection:Polynomial reproducing kernels}
A B-spline $\beta_P(t)$ of order $P$ is computed as the $(P+1)$-fold convolution of the box function $\beta_0(t)$ \cite{799930}:
\small
\begin{equation*}
\beta_P(t) =\underbrace{ \beta_0(t)*\beta_0(t)....*\beta_0(t)}_{P+1 \text{ times}},
\end{equation*}
\normalsize
where the anti-causal version of $\beta_0(t)$ is defined as:
\small
\begin{equation*}
\beta_0(t) = \begin{cases}
1, & -1\leq t\leq 0, \\
0, & \text{otherwise}.
\end{cases}
\end{equation*}
\normalsize

The B-spline of order $P$ satisfies the Strang-Fix conditions \cite{Strang2011} and hence, together with its uniform shifts, it can reproduce polynomials of maximum degree $P$:
\small
\begin{equation}\label{eq:poly_spline uniform}
\sum_{n \in \mathbb{Z}} c_{m,n} \beta_P(t-n)=t^m,
\end{equation}
\normalsize
where $m\in\{0,1,...,P\}$, and for a proper choice of the coefficients $c_{m,n}$.

For instance, the first-order B-spline satisfies Eq. (\ref{eq:poly_spline uniform}) for $P=1$, which means it can reproduce constant and linear polynomials, and is defined as:
\small
\begin{equation*}
\beta_1(t) = \begin{cases}
-t, & -1\leq t \leq 0,\\
2+t, & -2 \leq t < -1,\\
0, & \text{otherwise}.
\end{cases}
\end{equation*}
\normalsize

The first order B-spline has two continuous regions, each of which is a linear polynomial: $\beta_1^A(t) =-t$, for $t \in (-1,0)$ and $\beta_1^{B}(t)=2+t$, for $t \in (-2,-1)$. Using this observation, it is possible to show that the first-order B-spline, together with its \textit{non-uniformly} shifted versions can \textit{locally} reproduce polynomials of maximum degree $1$. In other words, it is possible to prove that within a time interval $I$ where the shifted kernels $\beta_1(t-t_n)$ have no knots, the following equation holds:
\small
\begin{equation}\label{eq:poly_spline non-uniform}
\sum_{n=0}^{N-1} c_{m,n}^I \beta_1(t-t_n)=t^m,
\end{equation}
\normalsize
where $N \geq 2$, $m \in \{0,1\}$, $t \in I$ and $\{t_n\}$ are non-uniform.

The proof can be outlined by setting $N=2$ for simplicity. Then, let $I$ be an interval where there are no knots of $\beta_1(t-t_0)$ and $\beta_1(t-t_1)$, with $I \subset (t_1-1,t_0)$. Furthermore, let $v_0(t)=\beta_1(t-t_0)=-t+t_0$ for $t \in I$ and $v_1(t)=\beta_1(t-t_1)=-t+t_1$ for $t \in I$. In the vector space of linear polynomials in $I$ which is a two-dimensional space, the elements $v_0(t)$ and $v_1(t)$ are linearly independent and so form a basis of the space, provided $t_0 \neq t_1$. Therefore, using a linear combination of the two functions, we can uniquely represent any vector in this space, including the vector $t$. In other words, we can determine the unique coefficients $c_{1,0}^I=\frac{t_1}{t_0-t_1}$ and $c_{1,1}^I=\frac{t_0}{t_1-t_0}$ that ensure $c_{1,0}^Iv_0(t)+c_{1,1}^Iv_1(t)=t$, for $t \in I$. Similarly, we find the unique coefficients $c_{0,0}^I=\frac{1}{t_0-t_1}$ and $c_{0,1}^I=\frac{1}{t_1-t_0}$ such that $c_{0,0}^Iv_0(t)+c_{0,1}^Iv_1(t)=1$, for $t \in I$. 
Hence, Eq. (\ref{eq:poly_spline non-uniform}) is satisfied in the knot-free interval $I$ for $N=2$. 

In the same manner, one can show that reproduction of constant and linear polynomials is achieved on any interval spanned by knot-free regions of at least two non-uniformly shifted B-splines. 
Lastly but importantly, for different knot-free intervals, the solution to Eq. (\ref{eq:poly_spline non-uniform}) differs, and this fact is highlighted in Fig.~\ref{fig:poly_recon}. Here, we depict two non-uniform shifts of the first-order B-spline, namely $\beta_1(t-2)$ and $\beta_1(t-2.625)$. The shifted kernel $\beta_1(t-2)$ has knots at $t=0$, $t=1$ and $t=2$, whilst $\beta_1(t-2.625)$ has knots at $t=0.625$, $t=1.625$ and $t=2.625$. As a result, reproduction of polynomials is possible within the knot-free regions $I_1=(0.625, 1)$ and $I_2=(1, 1.625)$, however with a different linear combination of the B-splines overlapping these regions, i.e. with $c_{m,n}^{I_1} \neq c_{m,n}^{I_2}$.

\begin{figure}[htb]
\vspace{-1em}
\centering
\includegraphics[width=0.4\textwidth]{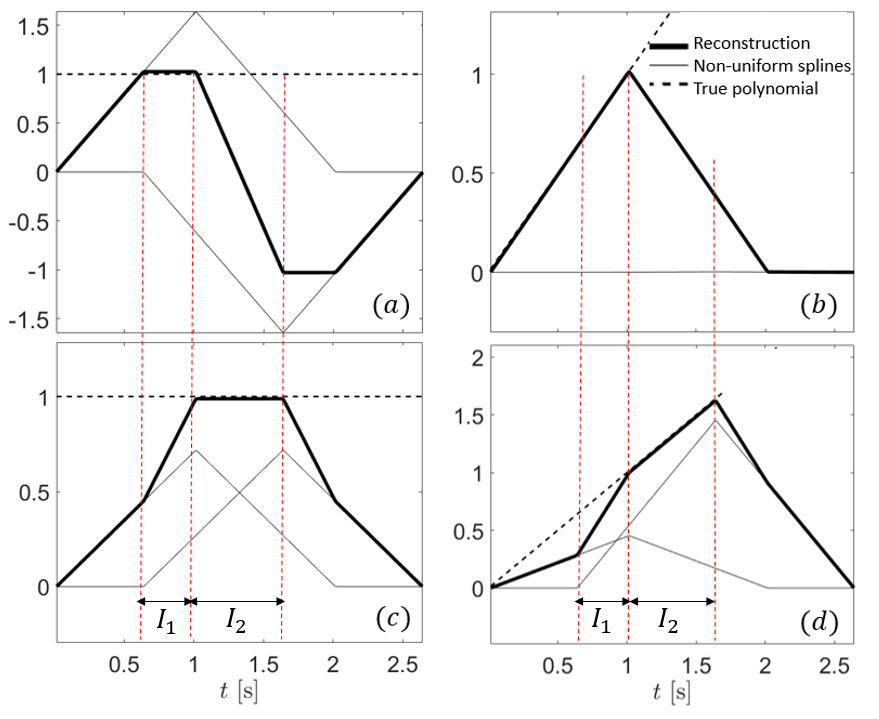}
\caption{Reproduction of constant and linear polynomials in two different time intervals, $I_1=(0.625, 1)s$ in $(a)$ and $(b)$, and $I_2=(1, 1.625)s$ in $(c)$ and $(d)$. In this case, two knot-free regions of two non-uniformly shifted first-order B-splines overlap $I_1$ and $I_2$.}
\label{fig:poly_recon}
\vspace{-1.5em}
\end{figure}

One can extend this result to the case of higher order polynomials by using B-splines of order $P>1$. This is due to the fact that polynomial splines are piecewise polynomial functions of degree $P$. Hence, in any interval $I$ that contains $P+1$ knot-free shifted versions of splines, it is possible to reproduce polynomials up to degree $P$.

\vspace{0.3em}
\subsubsection{Exponential splines}
\label{subsubsection: Exponential reproducing kernels}
The anti-causal version of the E-spline of first-order is defined as:
\small
\begin{equation*}
\varphi_1(t) = \begin{cases}
e^{-\alpha_0t}, & -1\leq t\leq 0, \\
0, & \text{otherwise}.
\end{cases}
\end{equation*}
\normalsize
where $\alpha_0$ can be either real or complex.

As with polynomial splines, E-splines of order $P$ are obtained from the convolution of first-order E-splines \cite{1408193}:
\small
\begin{equation}
\label{eq:higher_order_e_spline}
\varphi_P(t) =\varphi_{\alpha_0}(t)*\varphi_{\alpha_1}(t)....*\varphi_{\alpha_{P-1}}(t).
\end{equation}
\normalsize

An E-spline of order $P$ has compact support and can reproduce $P$ different exponentials of the form $e^{-\alpha_m t}$ \cite{1408193}:
\small
\begin{equation*}
\sum_{n \in \mathbb{Z}} c_{m,n} \varphi(t-n)=e^{-\alpha_m t},
\end{equation*} 
\normalsize
where $m=0,1,...,P$, and for a suitable choice of the coefficients $c_{m,n}$.

For example, the E-spline of order $P=2$ of support of arbitrary length $L$ is defined as:
\small
\begin{equation}
\varphi_2(t)= \begin{cases}
\label{eq: first_order_e_spline_definition_1}
\frac{e^{c_1-c_0}}{c_1-c_0}e^{-\alpha_0t}+\frac{e^{-c_1+c_0}}{c_0-c_1}e^{-\alpha_1t}, &\hspace{-0.5em}-L\leq t < -\frac{L}{2}, \\
\frac{1}{c_0-c_1}e^{-\alpha_0t}+\frac{1}{c_1-c_0}e^{-\alpha_1t}, &\hspace{-0.5em}-\frac{L}{2} \leq t \leq 0, \\
0, &\hspace{-0.5em} \text{otherwise},
\end{cases}
\end{equation}
\normalsize
where $\alpha_i \in \mathbb{C}$ (if $\Re\{\alpha_i\}=0$ then $\varphi_2(t) \in \mathbb{R}$), and where $c_i = \alpha_i \frac{L}{2}$ for $i=0,1$ in order to ensure continuity of $\varphi_2(t)$. Throughout the remainder of the paper, we assume for simplicity that $L=2$.

The second-order E-spline can reproduce the exponentials $e^{-\alpha_0t}$ and $e^{-\alpha_1t}$. In fact, we notice that within each of its knot-free regions, the function $\varphi_2(t)$ can be expressed as a linear combination of the exponentials $e^{-\alpha_0t}$ and $e^{-\alpha_1t}$.
This observation helps us prove that within any time interval $I$ which contains knot-free regions of non-uniformly shifted first-order E-splines, we can reproduce two exponentials:
\small
\begin{equation}\label{eq:exp_spline non-uniform}
\sum_{n=0}^{N-1} c_{m,n}^I \varphi_2(t-t_n)=e^{-\alpha_m t},
\end{equation}
\normalsize
where $N \geq2$, $m\in\{0,1\}$, $t\in I$ and $\{t_n\}$ are non-uniform. 

For example, let $I$ be an interval which contains knot-free regions of $\varphi_2(t-t_0)$ and $\varphi_2(t-t_1)$, with $I\subset(t_{1}-L,t_{0}-\frac{L}{2})$. Moreover, let $v_0(t)=\varphi_2(t-t_0)$ for $t\in I$ and $v_1(t)=\varphi_2(t-t_{1})$ for $t\in I$.
The elements $v_0(t)$ and $v_1(t)$ are linear combinations of $e^{-\alpha_0t}$ and $e^{-\alpha_1t}$, and therefore belong to the vector space spanned by these two exponentials. Moreover, $v_0(t)$ and $v_1(t)$ are linearly independent and so, form a basis of that vector space, since $t_{1} \neq t_0$.
Hence, using a linear combination of $v_0$ and $v_1$, we can uniquely represent any vector in this space, including $e^{-\alpha_0t}$ and $e^{-\alpha_1t}$. Therefore, in the interval $I$ where there are no knots, we can find unique coefficients $c_{m,0}^{I}$ and $c_{m,1}^{I}$ such that Eq. (\ref{eq:exp_spline non-uniform}) holds for $m\in \{0,1\}$.

Similarly, reproduction of two different exponentials is possible on any time interval spanned by knot-free regions of at least two shifted E-splines. Note that for different intervals $I_1$ and $I_2$, the solution to Eq. (\ref{eq:exp_spline non-uniform}) differs, i.e. $c_{m,n}^{I_1}\neq c_{m,n}^{I_2}$. This is highlighted in Fig.~\ref{fig:exp_recon}, where exponential reproduction is possible in the regions $I_1$ and $I_2$, but using a different linear combination of the E-splines that overlap these regions.
\begin{figure}[htb]
\vspace{-0.5em}
\centering
\includegraphics[width=0.42\textwidth]{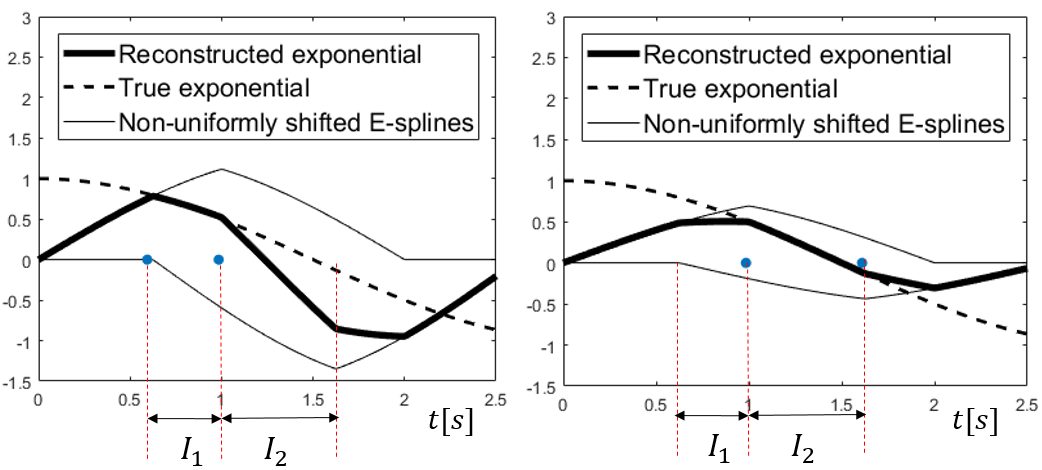}
\caption{Reproduction of $\Re\{e^{j\frac{2\pi}{5}t}\}$ in two different intervals, $I_1=(0.625, 1)s$ and $I_2=(1, 1.625)s$, overlapped by continuous regions of two non-uniformly shifted second-order E-splines.}
\label{fig:exp_recon}
\vspace{-0.5em}
\end{figure}

By using the same argument we can prove similar results for the general case of an E-spline of order $P$ and support of length $L$ which can reproduce $P$ different exponentials. Specifically, within an interval $I$ containing knot-free regions of at least $P$ non-uniformly shifted E-splines, we can reproduce $P$ different exponentials, such that Eq. (\ref{eq:exp_spline non-uniform}) holds for $N \geq P$ and $m \in \{0,1,...,P-1\}$. This is due to the fact that any knot-free interval of an E-spline of order $P$ is a linear combination of $P$ different exponentials.

Finally, let us consider the kernel $(\varphi_P*g)(t)$, where $\varphi_P(t)$ is a $P$-order E-spline which can reproduce the exponentials $e^{\alpha_mt}$, for $m=0,1,...,P-1$.
Furthermore, let us assume that $g(t)$ has compact support $[-\epsilon, \epsilon]$.
The support of $\varphi_P(t)$ is $[-L, 0]$ and its knots are located at instants $(-L + n \frac{L}{P})$ with $n \in \mathbb{N}$. Then, in the knot-free interval $(\frac{-L}{P},0)$ we can compactly represent  $\varphi_P(t)=\sum_{m=0}^{P-1} a_m e^{\alpha_mt}$, for some coefficients $a_m$.

If the length of the support of $g(t)$ satisfies $2\epsilon \leq \frac{L}{P}$ and $\int g(t)e^{-\alpha_mt}dt$ exists, then $(\varphi_P*g)(t)$ is given by:
\begin{equation} \label{eq: convo_pulses_1}
(\varphi_P*g)(t) = \sum_{m=0}^{P-1} a_mG_m e^{\alpha_mt},
\end{equation}
where $t \in (\epsilon -\frac{L}{P}, -\epsilon)$ and $G_m= \int_{-\epsilon}^{\epsilon} g(t)e^{-\alpha_mt}dt$.

Therefore, in the interval $(\epsilon -\frac{L}{P}, -\epsilon)$, $(\varphi_P*g)(t)$ is a linear combination of $P$ exponentials. As a result, within $I=(t_{N-1}+\epsilon-\frac{L}{P}, t_1- \epsilon)$, $(\varphi_P*g)(t)$ and its non-uniform shifts can reproduce $P$ exponentials, as follows:
\small
\begin{equation}
\label{eq:pulses_small_epsilon}
\sum_{n=0}^{N-1} c_{m,n}^I (\varphi_P*g)(t-t_n)=e^{-\alpha_m t},
\end{equation}
\normalsize
where $N \geq P$, and $m \in \{ 0,1,...,P-1\}$.

\section{Perfect Recovery of Signals from Timing Information obtained with a Crossing TEM}
\label{section:Perfect Recovery of Signals from Timing Information obtained with a Crossing Time Encoding Machine}
In the previous section, we showed how time encoding maps the input signal to a sequence of non-uniform samples, which depend on the signal and non-uniform shifts of the sampling kernel.
In what follows we assume that the sampling kernel $\varphi(t)$ is a second-order exponential reproducing spline, such that a linear combination of its non-uniformly shifted versions can reproduce two different exponentials, as described in Section \ref{subsubsection: Exponential reproducing kernels}. Moreover, $\varphi(t)$ has compact support of length $L$, with $\varphi(t)=0$ for $t \notin [-L, 0]$ and the two frequencies that this kernel can reproduce are $\alpha_0=j\omega_0$ and $\alpha_1=-\alpha_0$, which ensures that $\varphi(t)$ is a real-valued function. 

Under these assumptions, we study the problem of reconstructing different classes of non-bandlimited signals, from timing information obtained using the crossing TEM in Fig.~\ref{fig:comparator}. Specifically we present a method for estimation of an input Dirac. Here we show that two output spikes are sufficient to retrieve the input, provided they are located \textit{suitably} close to the Dirac, which is guaranteed by imposing conditions on the frequency and amplitude of the comparator's sinusoidal reference signal. 
We then extend this to retrieval of streams of Diracs and bursts of Diracs. While the reconstruction method proposed to retrieve one Dirac might not be unique, it has the advantage that it naturally generalizes to multiple Diracs.
We note that similar results could be proved using polynomial splines, but we omit these proofs to keep the focus of the paper on E-splines.
\vspace{-0.5em}
\subsection{Estimation of an Input Dirac}
\label{subsection:Estimation of an Input Dirac}
Let us consider a single input Dirac of the form:
\small
\begin{equation}
\vspace{-0.5em}
\label{eq:single input Dirac}
x(t)=x_1 \delta(t-\tau_1).
\end{equation}
\normalsize

\begin{prop}
\label{prop: Comparator 1 Dirac}
Let the sampling kernel $\varphi(t)$ be a second-order E-spline of support of length $L$, defined as in Eq. (\ref{eq: first_order_e_spline_definition_1}), with $\omega_1=-\omega_0$ and $0<\omega_0\leq\frac{\pi}{L}$. The filter $\varphi(t)$ and its non-uniform shifts can reproduce the exponentials $e^{j \omega_0 t}$ and $e^{j \omega_1 t}$ as in Eq. (\ref{eq:exp_spline non-uniform}). In addition, suppose that the reference signal $g(t)=A\cos(w_st)$ of the comparator in Fig.~\ref{fig:comparator} has amplitude $A>|x_1|$ and period $T_s<\frac{2L}{5}$.
Then, the timing information $\{t_1, t_2,...,t_N\}$ provided by the comparator TEM is a sufficient representation of an input Dirac as in Eq. (\ref{eq:single input Dirac}).
\end{prop}
\vspace{-1em}
\begin{proof}
From the timing information  $\{t_1, t_2\}$, we can retrieve the non-uniform output samples $y(t_1)$ and $y(t_2)$, as described in Eq. (\ref{eq:non-uniform samples comparator}).
In what follows we show that we can find a linear combination of the samples $y(t_1)$ and $y(t_2)$ to get $c_{m,1} y(t_1)+c_{m,2} y(t_2)=x_1 e^{j\omega_m \tau_1}$, for $m=0,1$, from which we can retrieve the input parameters $x_1$ and $\tau_1$.

For simplicity, suppose that the amplitude of the input Dirac satisfies $x_1>0$. In addition, the hypothesis that $ \varphi(t)$ reproduces $e^{\pm j \omega_0 t}$ with $0<\omega_0\leq\frac{\pi}{L}$ means that $0\leq\varphi(t)<1$, for $t \in [-\frac{L}{2}, 0]$. 
Then, since $0<x_1<A$, the output $y(t)=x_1 \varphi(\tau_1-t)$ of the crossing TEM satisfies $0\leq y(t)<A=\max(g(t))$, for $t \in [\tau_1, \tau_1+\frac{L}{2}]$. Since we assume $\frac{5T_s}{4}<\frac{L}{2}$, this means that $0\leq y(t)<A=\max(g(t))$, for $t \in [\tau_1, \tau_1+\frac{5T_s}{4}]$.

Let us then define the continuous function $h(t)=g(t)-y(t)$. Using Bolzano's intermediate value theorem \cite{Bolzano} and the fact that $0\leq y(t)<\max(g(t))$, we show that within the interval $(\tau_1, \tau_1+\frac{5T_s}{4}]$, the signal $h(t)$ crosses zero at least twice. In other words, $\exists t_1, t_2 \in (\tau_1, \tau_1+\frac{5T_s}{4}]$ such that $h(t_1)=h(t_2)=0$. For example, if we assume $h(\tau_1)=g(\tau_1)>0$, then $g(\tau_1+\frac{T_s}{2})<0$ and since $y(t)\geq 0$, we get $h(\tau_1+\frac{T_s}{2})=g(\tau_1+\frac{T_s}{2})-y(\tau_1+\frac{T_s}{2})<0$. Then, Bolzano's intermediate value theorem states that $\exists t_1 \in (\tau_1, \tau_1+\frac{T_s}{2}]$ such that $h(t_1)=0$. 

Using the same argument one can then show that $\exists t_2 \in (\tau_1+\frac{T_s}{2}, \tau_1+\frac{5T_s}{4}]$ such that $h(t_2)=0$. This follows from the assumption that $g(\tau_1)>0$, which implies that $\exists \epsilon \in [0, \frac{T_s}{2}]$ such that $g(\tau_1+\frac{3T_s}{4}+\epsilon)=\cos(\tau_1+\frac{3T_s}{4}+\epsilon)=\max(g(t))=A$, as highlighted in Fig.~\ref{fig:bolzano}. 
At the same time, we showed that $0\leq y(t)< A$ for $t\in [\tau_1, \tau_1+\frac{5T_s}{4}]$ and hence, we get $h(\tau_1+\frac{3T_s}{4}+\epsilon)=g(\tau_1+\frac{3T_s}{4}+\epsilon)-y(\tau_1+\frac{3T_s}{4}+\epsilon) > 0$. Since $h(\tau_1+\frac{T_s}{2})<0$ and $h(\tau_1+\frac{3T_s}{4}+\epsilon)>0$, Bolzano's intermediate value theorem guarantees that $\exists t_2 \in (\tau_1+\frac{T_s}{2}, \tau_1+\frac{3T_s}{4}+\epsilon]$ such that $h(t_2)=0$, for $\epsilon \in [0, \frac{T_s}{2}]$. Therefore, at the maximum value of $\epsilon$, we have proved that the second output spike satisfies $t_2 \in [\tau_1+\frac{T_s}{2}, \tau_1+\frac{5T_s}{4}]$.

\begin{figure}[htb]
\vspace{-0.7em}
\centering
\includegraphics[width=0.3\textwidth]{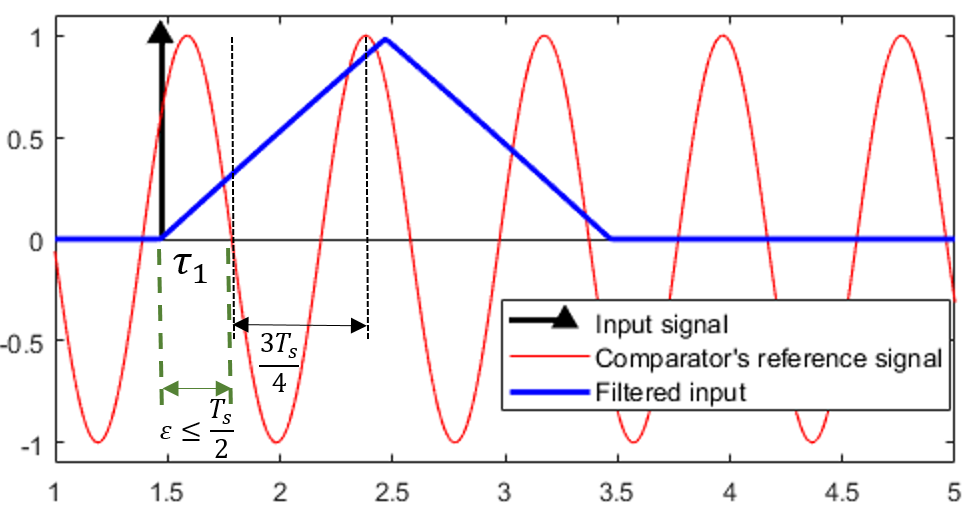}
\caption{Input Dirac located at $\tau_1$, filtered input and sinusoidal reference signal.}
\label{fig:bolzano}
\vspace{-1em}
\end{figure}

Hence, since we assume $\frac{L}{2}>\frac{5T_s}{4}$, we obtain the inequality $t_2\leq \tau_1+\frac{5T_s}{4}<\tau_1+\frac{L}{2}$. This guarantees that in a region $(\tau_1, \tau_1+\frac{L}{2})$ following a Dirac at $\tau_1$, there are at least 2 output samples, namely $y(t_1)$ and $y(t_2)$, as depicted in Fig.~\ref{fig:condition_comp_1_dirac_with_times}.

Then, in the interval $I=(t_2-\frac{L}{2}, t_1)$, which does not contain knots of either $\varphi(t-t_1)$ or $\varphi(t-t_2)$, we can reproduce two exponentials as described in Section \ref{subsubsection: Exponential reproducing kernels}. Specifically, we can find coefficients $c_{m,n}^I$ such that:
\vspace{-0.8em}
\small
\begin{equation}
\label{eq:exp recon one dirac}
\sum_{n=1}^{2} c_{m,n}^I \varphi(t-t_n) = e^{j\omega_m t}, \text{for } m\in\{0,1\}.
\vspace{-0.5em}
\end{equation}
\normalsize

We then define the signal moments $s_m$ as follows:
\small
\begin{equation}
\vspace{-0.5em}
\label{eq:moments one dirac}
\begin{split}
s_m &= \sum_{n=1}^{2} c_{m,n}^I y(t_n) \aeq  \sum_{n=1}^{2} c_{m,n}^I \langle x(t), \varphi(t-t_n)\rangle\\
&\beq \int_{-\infty}^{\infty} x(t) \sum_{n=1}^{2} c_{m,n}^I  \varphi(t-t_n) dt \\
&\ceq \int_{-\infty}^{\infty}  x_1 \delta(t-\tau_1)  \sum_{n=1}^{2} c_{m,n}^I  \varphi(t-t_n) dt \\
&\deq \int_I x_1 \delta(t-\tau_1)  e^{j\omega_m t}dt =  x_1 e^{j\omega_m \tau_1}=b_1 u_1^m,
\end{split}
\vspace{-1.5em}
\end{equation}
\normalsize
\vspace{-0.3em}
where $b_1 \defeq x_1 e^{j \omega_0 \tau_1}$, $u_1 \defeq e^{j\lambda \tau_1}$, the frequencies $\omega_m = \omega_0 + \lambda m$, for $m\in\{0,1\}$, and $\lambda=-2\omega_0$.

\begin{figure}[htb]
\vspace{-0.7em}
\centering
\includegraphics[width=0.35\textwidth]{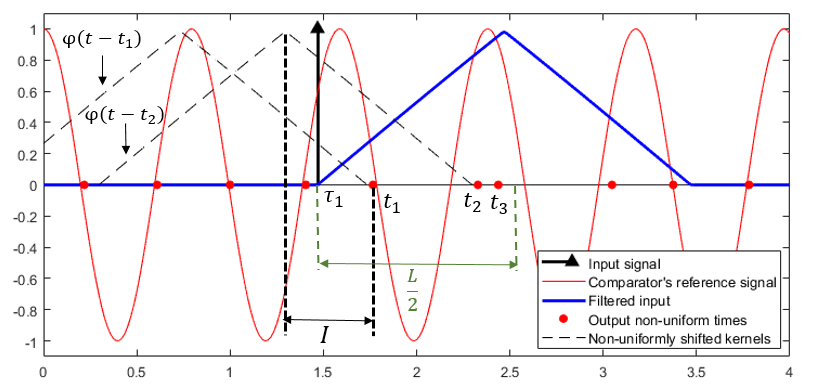}
\caption{Time encoding of an input Dirac located at $\tau_1 \in I$, when $\frac{L}{2}>\frac{5T_s}{4}$. In the interval $I$, $\varphi(t-t_1)$ and $\varphi(t-t_2)$ have no knots.}
\label{fig:condition_comp_1_dirac_with_times}
\vspace{-1em}
\end{figure}

In these derivations, $(a)$ follows from Eq. (\ref{eq:non-uniform samples comparator}), $(b)$ from the linearity of the inner product, and $(c)$ from Eq. (\ref{eq:single input Dirac}). Moreover, $(d)$ holds since $t_1,t_2 \in (\tau_1, \tau_1+\frac{L}{2})$, which means $\tau_1 \in I$, and from the local exponential reproduction property of $\varphi(t)$ in the region $I$, as given in Eq. (\ref{eq:exp recon one dirac}). 

The unknowns $\{b_1,u_1\}$ can be uniquely retrieved from the signal moments, as follows: $b_1 = s_0$ and $u_1 = \frac{s_1}{s_0}$. 
More generally, the parameters $\{b_1,u_1\}$ can also be found using the annihilating filter method \cite{Stoica104835}, also known as Prony's method \cite{Prony} (see Appendix \ref{appendix:Prony's method}). Then, we get the Dirac's amplitude and location, using $b_1 = x_1 e^{j \omega_0 \tau_1}$ and $u_1 = e^{j\lambda \tau_1}$.
\end{proof}

\vspace{-1em}
\subsection{Estimation of a Stream of Diracs}
Let us now consider the case of a stream of Diracs:
\small
\begin{equation}
\vspace{-0.7em}
\label{eq:input stream Diracs}
x(t) = \sum_{k}x_k \delta(t-\tau_k).
\vspace{-0.1em}
\end{equation}
\normalsize

\begin{prop}
\label{prop: Comparator Stream Diracs}
Let the sampling kernel $\varphi(t)$ be a second-order E-spline of support of length $L$, defined as in Eq. (\ref{eq: first_order_e_spline_definition_1}), with $\omega_1=-\omega_0$ and $0<\omega_0\leq\frac{\pi}{L}$. The filter $\varphi(t)$ and its non-uniform shifts can reproduce the exponentials $e^{j \omega_0 t}$ and $e^{j \omega_1 t}$ as in Eq. (\ref{eq:exp_spline non-uniform}). In addition, suppose that the reference signal $g(t)=A\cos(w_st)$ of the comparator has amplitude $A>|x_k|$, $\forall k$ and period $T_s<\frac{2L}{5}$, and that the minimum spacing between consecutive Diracs is larger than $L$.
Then, the timing information $\{t_1, t_2, ..., t_N\}$ provided by the device shown in Fig.~\ref{fig:comparator} is a sufficient representation of a stream of Diracs as in Eq. (\ref{eq:input stream Diracs}).
\end{prop}
\begin{proof}
The input stream of Diracs can be sequentially estimated as follows. The first Dirac $x_1 \delta(t-\tau_1)$ can be uniquely estimated using the first two non-zero samples $y(t_1)$ and $y(t_2)$, as presented in Section \ref{subsection:Estimation of an Input Dirac}. Once we know $\tau_1$, we retrieve the first two non-zero samples $y(t_n)$ and $y(t_{n+1})$ located after $\tau_1+L$, and use these to estimate the second Dirac $x_2 \delta(t-\tau_2)$. We then sequentially retrieve the next Dirac using the first two non-zero samples located after $\tau_2+L$, as illustrated in Fig.~\ref{fig:comp_iter}. 
In what follows we show that once $x_1\delta(t-\tau_1)$ has been estimated, we can use $y(t_n)$ and $y(t_{n+1})$ to estimate the second Dirac in the stream. Since we assume that the separation between input Diracs is larger than the length $L$ of the kernel's support, then the location $\tau_2$ of the second Dirac satisfies $\tau_1+L<\tau_2<t_n$.
Moreover, provided the period of the comparator's signal satisfies $T_s<\frac{2L}{5}$, Bolzano's intermediate value theorem \cite{Bolzano} guarantees that $y(t_n),y(t_{n+1}) \in (\tau_2, \tau_2+\frac{L}{2})$, as previously outlined in Section \ref{subsection:Estimation of an Input Dirac}. Then, the interval $I=(t_{n+1}-\frac{L}{2}, t_n)$ contains no knots of either $\varphi(t-t_n)$ or $\varphi(t-t_{n+1})$, and perfect exponential reproduction can be achieved. Hence we can compute the signal moments using similar derivations as in Eq. (\ref{eq:moments one dirac}):
\small
\begin{equation*}
s_m = c_{m,n}^I y(t_n) +  c_{m,n+1}^I y(t_{n+1}) = x_2 e^{j\omega_m \tau_2}.
\end{equation*}
\normalsize

Finally, we can estimate $x_2$ and $\tau_2$ from $s_m$, using Prony's method.
Once the second Dirac has been estimated, we use subsequent non-uniform output samples after $\tau_2+L$ in order to sequentially retrieve the next Diracs.
\end{proof}
\vspace{-0.5em}
The time encoding  of the stream of Diracs is depicted in Fig.~\ref{fig:comp_iter}. Here, the filter is a second-order E-spline, of support length $L=2$, which can reproduce the exponentials $e^{\pm j \frac{\pi}{3}t}$. The frequency of the comparator's test signal is $f_s=1.26>\frac{5}{2L}$ and the separation between Diracs is at least $L=2$. The amplitudes and locations of the estimated Diracs are exact to numerical precision.

\begin{figure}[htb]
\vspace{-1em}
\centering
\includegraphics[width=0.32\textwidth]{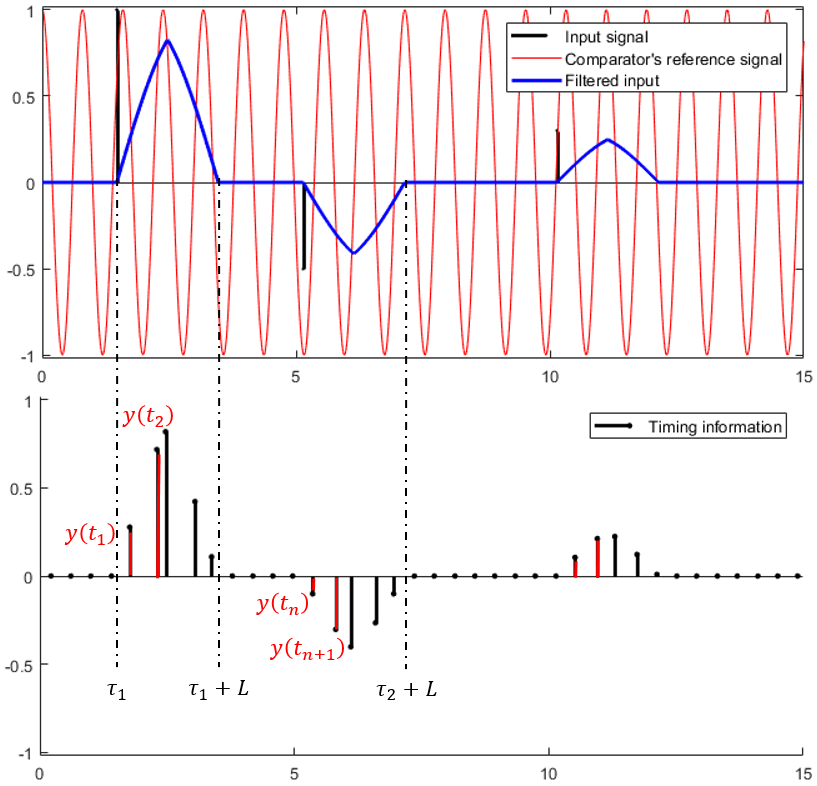}
\caption{Sampling of a stream of Diracs using the crossing TEM. The input signal, filtered input and sinusoidal reference signal are depicted in the top plot, and the output non-uniform samples in the bottom plot. The output samples used to retrieve each Dirac in the input stream are highlighted in red.}
\label{fig:comp_iter}
\vspace{-1.5em}
\end{figure}

\addtocounter{footnote}{-1}\let\thefootnote\svthefootnote
\subsection{Multi-channel Estimation of Bursts of Diracs}
Let us now consider a sequence of bursts of $K$ Diracs:
\small
\begin{equation}
\label{eq:input burst Diracs}
x(t) = \sum_b \sum_{k=1}^K x_{b,k}\delta(t-\tau_{b,k}),
\end{equation}
\normalsize
where the amplitudes $x_{b,k}$ in the same burst $b$ have the same sign and satisfy $|x_{b,k}|<A_{max}$.

\begin{prop}
\label{prop: Comparator Bursts Diracs}
Let us consider a system of $M \geq K$ TEM devices as in Fig.~\ref{fig:comparator}.The filter $\varphi(t)$ of the $m^{th}$ TEM is a second-order E-spline whose support has length $L$, and which can reproduce two different exponentials, $e^{j\omega_{m_0}t}$ and $e^{j\omega_{m_1}t}$, with $\omega_{m_0}=\omega_{0}+\lambda m$, $\lambda=\frac{-2\omega_0}{2M-1}$, $0<\omega_{0} \leq \frac{\pi}{L}$, $\omega_{m_1}=-\omega_{m_0}$, and $m=0,1,...,M-1$.
Furthermore, suppose the reference signal $g(t)=A\cos(w_st)$ has amplitude $A>KA_{max}$ and period $T_s<\frac{2L}{7}$.  In addition, let us assume the spacing between consecutive bursts is larger than $L$, and the maximum separation between the last and first Dirac in any burst $b$ satisfies $\tau_{b,K}-\tau_{b,1}< \frac{T_s}{2}$. 
Then, the timing information $t_{1,m},t_{2,m},...,t_{N,m}$ for $m=0,1,...,M-1$ provided by $M$ devices as in Fig.~\ref{fig:comparator} is a sufficient representation of bursts of $K$ Diracs as in Eq. (\ref{eq:input burst Diracs}).

\end{prop}
\begin{proof}
See Appendix \ref{Appendix: Comparator Bursts Diracs}.
\end{proof}

We summarize the results in this section by showing possible choices of the hyperparameters of the crossing TEM, and how they influence the density of output samples.
This relationship is presented in Table \ref{table:table4}, for the case of a sequence of bursts of $K$ Diracs. Here, $M$ is the minimum number of channels, $P$ is the order of the sampling kernel for each channel, $L$ is the length of the support of the sampling kernel and $f_s^{min}$ the minimum frequency of the comparator's sinusoidal reference. The table shows both the average sample density, as well as the ideal sample density\footnote{The ideal sample density is computed under the assumption that two samples are necessary to reconstruct one Dirac and that on average there are $K$ Diracs in an interval $L+\epsilon$ with $\epsilon>0$.} required for perfect estimation of each of the bursts of $K$ Diracs.

\small
\begin{table}[h]
\begin{center}
\caption{Choice of hyperparameters of the crossing TEM for reconstructing bursts of $K$ Diracs.
 \label{table:table4}}
\setlength\extrarowheight{5pt} 
\hspace{-0.5em}\begin{tabular}{|c|c|c|c|P{1.7cm}|P{1.45cm}|P{1.45cm}|} 
        \hline
      $K$ & $M$ & $P$ & $L$ & $f_s^{min}$ & Ideal sample density (samples/s)& Sample density (samples/s)\\
      \hline
      1 & 1 &  $2$ & $2$\scriptsize{ sec} & $\frac{5}{2L}=1.25$\scriptsize{ Hz} & $\frac{2}{L+\epsilon}$& $2f_s = 2.5$  \\
        \hline
      2 & 2& $2$ & $2$\scriptsize{ sec}  & $\frac{7}{2L}=1.75$\scriptsize{ Hz} &  $\frac{4}{L+\epsilon}$& $2Mf_s =7$  \\
        \hline
      4 & 4& $2$ & $2$\scriptsize{ sec}  & $\frac{7}{2L}=1.75$\scriptsize{ Hz} &  $\frac{8}{L+\epsilon}$& $2Mf_s =14$  \\
        \hline
\end{tabular}
\end{center}

\vspace{-1em}
\end{table}
\normalsize

We conclude this section by making the observation that the number of redundant samples of the crossing TEM is large, since samples are recorded even when the input is zero. In this case, output samples are recorded at the time instants when the sinusoidal reference signal crosses zero.
In what follows, we aim to use the same decoding framework, however with a more efficient acquisition device, the integrate-and-fire TEM.

\section{Perfect Recovery from Timing Information obtained with an Integrate-and-fire TEM}
\label{section:Perfect Recovery of Signals from Timing Information obtained with an Integrate-and-fire System}
We now shift our focus on the integrate-and-fire TEM in Fig.~\ref{fig:integrate_and_fire_model_2}.
In particular, we show how to perfectly estimate an input Dirac, and extend this method to streams and bursts of Diracs, streams of pulses as well as piecewise constant signals. The retrieval of these signals from their timing information is perfect, provided the threshold of the trigger comparator is small enough to ensure a sufficient density of output samples.
As it will become evident in Section \ref{section:Density of Non-uniform Samples}, an important feature of the integrate-and-fire model is that it can be more efficient than the comparator or a system based on uniform sampling, in the case of input signals with a small number of Diracs, because it leads to a smaller number of samples.
\vspace{-1.2em}
\subsection{Estimation of an Input Dirac}
\label{subsection:Estimation of an Input Dirac Integrator}
\begin{prop}
\label{prop:Estimation of an Input Dirac Integrator}
Let the sampling kernel $\varphi(t)$ be a second-order E-spline of support of length $L$, defined as in Eq. (\ref{eq: first_order_e_spline_definition_1}), with $\omega_1=-\omega_0$ and $0<\omega_0\leq\frac{\pi}{L}$. The filter $\varphi(t)$ and its non-uniform shifts can reproduce the exponentials $e^{j \omega_0 t}$ and $e^{j \omega_1 t}$ as in Eq. (\ref{eq:exp_spline non-uniform}). In addition, suppose that the trigger mark of the comparator satisfies:
\vspace{-0.5em}
\small
\begin{equation*}
\vspace{-0.5em}
0<C_T <\frac{A_{min}}{3}\int_{0}^{\frac{L}{2}}\varphi(-t) dt,
\end{equation*}
\normalsize
where $A_{min}$ is the absolute minimum amplitude of the Dirac.
Then, the timing information $\{t_1,t_2,...,t_N\}$ provided by the integrate-and-fire TEM in Fig.~\ref{fig:integrate_and_fire_model_2} is a sufficient representation of an input Dirac as in Eq. (\ref{eq:single input Dirac}).
\end{prop}

\vspace{-0.5em}
\begin{proof}
We will prove that the upper bound on $C_T$ guarantees that the integrated filtered input $y(t)=x_1\varphi(\tau_1-t)$ reaches the trigger mark at least three times in the interval $(\tau_1, \tau_1+\frac{L}{2})$. We will then show how we can use the second and third output samples $y(t_2)$ and $y(t_3)$ to perfectly estimate the input Dirac, given that the integrated filtered input has no discontinuities in the interval $(\tau_1, \tau_1+\frac{L}{2})$.

First, we note that:
\small
\begin{equation}
\label{eq:rewriting the e-spline of order 1}
\int_{0}^{\frac{L}{2}}\varphi(-t) dt=\int_{\tau_1}^{\tau_1+\frac{L}{2}} \varphi(\tau_1-t)dt\aeq \frac{1}{\omega_0^2}[1-\cos(\omega_0\frac{L}{2})],
\end{equation}
\normalsize
where $(a)$ follows from Eq. (\ref{eq: first_order_e_spline_definition_1}), given $\alpha_0=-j\omega_0$, $\alpha_1=-j\omega_1$ and $\omega_1=-\omega_0$.

Then, we assume for simplicity that the Dirac's amplitude satisfies $x_1>0$ and re-write the upper bound on $C_T$ as:
\small
\begin{equation}
\label{eq:threshold_cond_1_dirac expanded}
3C_T <A_{\min}\int_{0}^{\frac{L}{2}}\varphi(-t) \bl \int_{\tau_1}^{\tau_1+\frac{L}{2}} x_1\varphi(\tau_1-t)dt,
\end{equation}
\normalsize
where $(b)$ follows from Eq. (\ref{eq:rewriting the e-spline of order 1}).

Furthermore, from Eq. (\ref{eq:non_uniform_samples integrator}) and (\ref{eq:non_uniform_sample_1 integrator}), we know that:
\small
\begin{equation}
\label{eq:threshold_cond_1_dirac expanded 2}
\begin{split}
3C_T  &= \int_{\tau_1}^{t_3} f(t)dt \ceq \int_{\tau_1}^{t_3}  x_1\varphi(\tau_1-t)dt,
\end{split}
\end{equation}
\normalsize
where $(c)$ follows from Eq. (\ref{eq:filtered input integrator}) and given the input signal is $x(t)=x_1\delta(t-\tau_1)$.

Then, from Eq. (\ref{eq:threshold_cond_1_dirac expanded}) and Eq. (\ref{eq:threshold_cond_1_dirac expanded 2}), we obtain the inequality:
\small
\begin{equation}
\label{eq:integrator 1 dirac first 3 samples}
\int_{\tau_1}^{t_3}  x_1\varphi(\tau_1-t)dt <\int_{\tau_1}^{\tau_1+\frac{L}{2}} x_1\varphi(\tau_1-t)dt.
\end{equation}
\normalsize

Using the hypothesis $\omega_1=-\omega_0$, together with Eq. (\ref{eq: first_order_e_spline_definition_1}), we obtain $\varphi(\tau_1-t) = \frac{\sin(\omega_0(t-\tau_1))}{\omega_0}$, for $t \in (\tau_1, \tau_1+\frac{L}{2})$. Given the assumption $0 < \omega_0\leq \frac{\pi}{L}$, we get $0 \leq \omega_0(t-\tau_1) \leq \frac{\pi}{2}$ and therefore, $\sin(\omega_0(t-\tau_1))>0$ for $t \in (\tau_1, \tau_1+\frac{L}{2})$. Hence, since $\varphi(\tau_1-t)$ is positive in the range $(\tau_1, \tau_1+\frac{L}{2})$ and using Eq. (\ref{eq:integrator 1 dirac first 3 samples}), we get that $t_3<\tau_1+\frac{L}{2}$.

As a result, the locations of the first non-uniform output samples satisfy $t_1, t_2,t_3 \in(\tau_1, \tau_1+\frac{L}{2})$, and can be computed using Eq. (\ref{eq: output_sample_1_equivalent_filter}) and Eq. (\ref{eq: output_samples_equivalent_filter}) as follows:
\small
\begin{equation*}
y(t_1) = \int_{\tau_1}^{t_1} f(t) dt=\langle x(t), (\varphi*q_{\theta_1})(t-\tau_1) \rangle,
\end{equation*}
\normalsize
\small
\begin{equation}
\label{eq:output_sample_2_for_1_dirac}
\vspace{-0.5em}
y(t_2) = \langle x(t), (\varphi*q_{\theta_2})(t-t_1) \rangle,
\end{equation}
\normalsize
\small
\begin{equation}
\label{eq:output_sample_3_for_1_dirac}
\vspace{-0.5em}
y(t_3) = \langle x(t), (\varphi*q_{\theta_3})(t-t_2) \rangle,
\end{equation}
\normalsize
for $\theta_1=t_1-\tau_1$, $\theta_2=t_2-t_1$ and $\theta_3=t_3-t_2$.

Furthermore, since $\varphi(t)$ is a second-order E-spline which can reproduce the exponentials $e^{j\omega_0t}$ and $e^{j\omega_1t}$ as in Eq. (\ref{eq: first_order_e_spline_definition_1}), and given the definition of $q_{\theta_n}(t)$ in Eq. (\ref{eq:box function integrator}), we have that:
\small
\begin{equation*}
\begin{split}
(\varphi*q_{\theta_1})(t-\tau_1) &= \frac{1}{\omega_0(\omega_0-\omega_1)}[(e^{-j\omega_0 t_1}-e^{-j\omega_0\tau_1})e^{j\omega_0 t}\\
&+(e^{-j\omega_1 t_1}-e^{-j\omega_1\tau_1})e^{j\omega_1 t}],
\end{split}
\end{equation*}
\normalsize
for $t \in (t_1-\frac{L}{2},t_1)$.

Similarly:
\vspace{-0.5em}
\small
\begin{equation*}
\begin{split}
(\varphi*q_{\theta_2})(t-t_1) &= \frac{1}{\omega_0(\omega_0-\omega_1)}[(e^{-j\omega_0 t_2}-e^{-j\omega_0t_1})e^{j\omega_0 t}\\
&+(e^{-j\omega_1 t_2}-e^{-j\omega_1t_1})e^{j\omega_1 t}],
\end{split}
\end{equation*}
\normalsize
for $t \in (t_2-\frac{L}{2},t_2)$, and
\small
\begin{equation*}
\begin{split}
(\varphi*q_{\theta_3})(t-t_2) &= \frac{1}{\omega_0(\omega_0-\omega_1)}[(e^{-j\omega_0 t_3}-e^{-j\omega_0 t_2})e^{j\omega_0 t}\\
&+(e^{-j\omega_1 t_3}-e^{-j\omega_1 t_2})e^{j\omega_1 t}],
\end{split}
\end{equation*}
\normalsize
for $t \in (t_3-\frac{L}{2},t_3)$.

The shifted kernel $(\varphi*q_{\theta_1})(t-\tau_1)$ depends on the Dirac's location $\tau_1$, and hence its shape cannot be determined a-priori. On the other hand, the shifted kernels $(\varphi*q_{\theta_2})(t-t_1)$ and $(\varphi*q_{\theta_3})(t-t_2)$ are independent of $\tau_1$ and can be written as a linear combination of the exponentials $e^{j\omega_0 t}$ and $e^{j\omega_1 t}$, for $t \in (t_3-\frac{L}{2},t_1)$. Therefore, in the interval $I=(t_3-\frac{L}{2},t_1)$, where there are no knots of either the shifted kernel $(\varphi*q_{\theta_2})(t-t_1)$ or $(\varphi*q_{\theta_3})(t-t_2)$, we can use the proof in Section \ref{subsubsection: Exponential reproducing kernels} to find the unique coefficients $c_{m,2}^I$ and $c_{m,3}^I$ such that:
\vspace{-0.5em}
\small
\begin{equation}
\label{eq:exp recon modified kernel integrator}
\sum_{n=2}^3c_{m,n}^I(\varphi*q_{\theta_n})(t-t_{n-1})= e^{j\omega_m t}, 
\end{equation}
\normalsize
for $m \in \{0,1\}$ and $t \in (t_3-\frac{L}{2},t_1)$.

Then, we can define the signal moments as:
\small
\begin{equation}
\label{eq:integrator_moments_single_dirac}
\begin{split}
s_m &= \sum_{n=2}^3 c_{m,n}^Iy(t_n) \deq x_1 \sum_{n=2}^3c_{m,n}^I(\varphi*q_{\theta_n})(\tau_1-t_{n-1})\\
&\eeq x_1 e^{j\omega_m\tau_1}, \text{ for } m\in\{0,1\}.
\end{split}
\end{equation}
\normalsize

In the derivations above, $(d)$ follows from Eq. (\ref{eq:single input Dirac}), (\ref{eq:output_sample_2_for_1_dirac}) and (\ref{eq:output_sample_3_for_1_dirac}), and $(e)$ follows from $\tau_1 \in (t_3-\frac{L}{2}, t_1)$ which is true given Eq. (\ref{eq:integrator 1 dirac first 3 samples}), and since the property in Eq. (\ref{eq:exp recon modified kernel integrator}) holds within $(t_3-\frac{L}{2}, t_1)$.
Finally, using Prony's method we can uniquely estimate parameters $x_1$ and $\tau_1$, from the two signal moments $s_m$ given by Eq. (\ref{eq:integrator_moments_single_dirac}), for $m\in \{0,1\}$ and $\omega_1=-\omega_0$.
\vspace{-0.5em}
\end{proof}
\vspace{-1em}
\subsection{Estimation of a Stream of Diracs}
\label{subsection:Estimation of a Stream of Diracs Integrator}
\begin{prop} \label{prop:Estimation of a Stream of Diracs Integrator}
Let the sampling kernel $\varphi(t)$ be a second-order E-spline of support of length $L$, defined as in Eq. (\ref{eq: first_order_e_spline_definition_1}), with $\omega_1=-\omega_0$ and $0<\omega_0\leq\frac{\pi}{L}$. The filter $\varphi(t)$ and its non-uniform shifts can reproduce the exponentials $e^{j \omega_0 t}$ and $e^{j \omega_1 t}$ as in Eq. (\ref{eq:exp_spline non-uniform}). In addition, assume that the minimum separation between consecutive Diracs is $L$ and the trigger mark of the comparator satisfies:
\small
\begin{equation}
\label{eq:threshold_cond_stream_diracs}
0<C_T <\frac{A_{\min}}{4\omega_0^2}[1-\cos(\omega_0\frac{L}{2})],
\end{equation}
\normalsize
where $A_{\min}$ is the absolute minimum amplitude of any Dirac in the input signal.

Then, the timing information $\{t_1, t_2, ..., t_N\}$ provided by the integrate-and-fire TEM in Fig.~\ref{fig:integrate_and_fire_model_2} is a sufficient representation of a stream of Diracs as in Eq. (\ref{eq:input stream Diracs}).
\end{prop}

\begin{proof}
The first Dirac $\delta_1=x_1 \delta(t-\tau_1)$ can be correctly estimated using the method in Section \ref{subsection:Estimation of an Input Dirac Integrator}, since Eq. (\ref{eq:threshold_cond_stream_diracs}) satisfies the requirements of Proposition \ref{prop:Estimation of an Input Dirac Integrator}. Then, suppose we aim to estimate the second Dirac in the input signal, and let us assume for simplicity that its amplitude satisfies $x_2>0$. Moreover, let us denote the output spike locations in the interval $(\tau_1, \tau_1+L)$ with $t_1,t_2,...,t_{n-1}$, and the time information after $\tau_1+L$ with $t_n, t_{n+1},...,t_N$. Then, given the hypothesis that the minimum separation between consecutive Diracs is $L$, the location of the second Dirac must satisfy $\tau_2 \in (\tau_1+L, t_n)$. We also have that:
\small
\begin{equation*}
\int_{\tau_2}^{\tau_2+\frac{L}{2}}f(\tau) d\tau= \int_{\tau_2}^{\tau_2+\frac{L}{2}} x_2\varphi(\tau_2-\tau)d\tau \aeq \frac{x_2}{\omega_0^2}[1-\cos(\omega_0\frac{L}{2})],
\end{equation*}
\normalsize
where $(a)$ follows from Eq. (\ref{eq: first_order_e_spline_definition_1}), for $\omega_1=-\omega_0$.

This shows the upper bound in Eq. (\ref{eq:threshold_cond_stream_diracs}) is equivalent to:
\small
\begin{equation}
\label{eq:threshold_cond_stream_diracs upper}
4C_T <\int_{\tau_2}^{\tau_2+\frac{L}{2}}f(\tau) d\tau.
\end{equation}
\normalsize

Furthermore, we have that:
\small
\begin{equation}
\label{eq:stream_diracs_output_spikes}
\begin{split}
\int_{\tau_2}^{t_{n+2}}f(\tau) d\tau &=\int_{t_{n-1}}^{t_{n+2}}f(\tau)d\tau - \int_{t_{n-1}}^{\tau_2}f(\tau)d\tau \\
&\beq 3C_T - \int_{t_{n-1}}^{\tau_2}f(\tau)d\tau \cl 4C_T,
\end{split}
\end{equation}
\normalsize
where $(b)$ follows from Eq. (\ref{eq:non_uniform_samples integrator}), and $(c)$ holds since $t_{n-1}$ and $t_n$ are consecutive output spikes, and $t_n>\tau_2>t_{n-1}$.

As a result, Eq. (\ref{eq:threshold_cond_stream_diracs upper}) and (\ref{eq:stream_diracs_output_spikes}) give the following inequality:
\small
\begin{equation}
\label{eq:stream_diracs_output_spikes2}
\int_{\tau_2}^{t_{n+2}}f(\tau) d\tau<\int_{\tau_2}^{\tau_2+\frac{L}{2}}f(\tau) d\tau.
\end{equation}
\normalsize

As shown in Section \ref{subsection:Estimation of an Input Dirac Integrator}, the sampling kernel satisfies $\varphi(t)>0$ for $x_2>0$, within the interval $(\tau_2,\tau_2+\frac{L}{2})$. This means that the inequality in Eq. (\ref{eq:stream_diracs_output_spikes2}) is equivalent to $t_{n+2}<\tau_2+\frac{L}{2}$, which guarantees that the output samples  $y_{n}$, $y_{n+1}$ and $y_{n+2}$ occur in the time interval $(\tau_2, \tau_2+\frac{L}{2})$.
Using the model of Fig.~\ref{fig:integrate_and_fire_model_2}, we compute these non-uniform output samples as:
\small
\begin{equation*}
y(t_n)=y_n=\int_{t_{n-1}}^{\tau_1+L} x_1 \varphi(\tau_1-\tau) d\tau+\int_{\tau_2}^{t_n} x_2 \varphi(\tau_2-\tau)d\tau,
\end{equation*}
\normalsize
\small
\begin{equation*}
y(t_{n+1})=y_{n+1}=\int_{t_{n}}^{t_{n+1}} x_2 \varphi(\tau_2-\tau)d\tau,
\end{equation*}
\normalsize
\small
\begin{equation*}
y(t_{n+2})=y_{n+2}=\int_{t_{n+1}}^{t_{n+2}} x_2 \varphi(\tau_2-\tau)d\tau.
\vspace{-0.5em}
\end{equation*}
\normalsize

The sample $y(t_n)$ contains information of both $\delta_1$ and $\delta_2$, and hence cannot be used for estimation of the latter Dirac. On the other hand, since $t_{n+1}, t_{n+2} \in(\tau_2, \tau_2+\frac{L}{2})$, we can use the samples $y_{n+1}$ sand $y_{n+2}$ to compute the signal moments as in Section \ref{subsection:Estimation of an Input Dirac Integrator}:
\small
\begin{equation*}
s_m = c_{m,1} y_{n+1} + c_{m,2} y_{n+2} = x_2 e^{j\omega_m \tau_2}, \text{ for } m \in \{0,1\}.
\end{equation*}
\normalsize

Once $\delta_2$ is estimated from $s_m$ using Prony's method, we use the non-uniform output samples after $\tau_2+L$, in order to sequentially retrieve the next Diracs in the input signal.
\end{proof}
\vspace{-0.5em}
The sampling and reconstruction of a stream of $K=3$ Diracs of minimum absolute amplitude $A_{\min}=1$ are depicted in Fig.~\ref{fig:int_stream_diracs}. Here, the filter is a second-order E-spline, of support of length $L=2$, which can reproduce the exponentials $e^{\pm j \frac{\pi}{3}t}$, and the comparator's trigger mark is $C_T=0.11$, which satisfies Eq. (\ref{eq:threshold_cond_stream_diracs}). 
Fig.~\ref{fig:int_stream_diracs}(b) shows the filtered input and the output of the integrator. The amplitudes and locations of the estimated Diracs are exact to numerical precision. Finally, in Fig.~\ref{fig:int_stream_diracs}(c) we observe that there are no output spikes in a region where the input signal is constant (zero), which leads to small average density of samples. 
\vspace{-0.5em}
\begin{figure}[htb]
\vspace{-0.5em}
\centering
\includegraphics[width=0.4\textwidth]{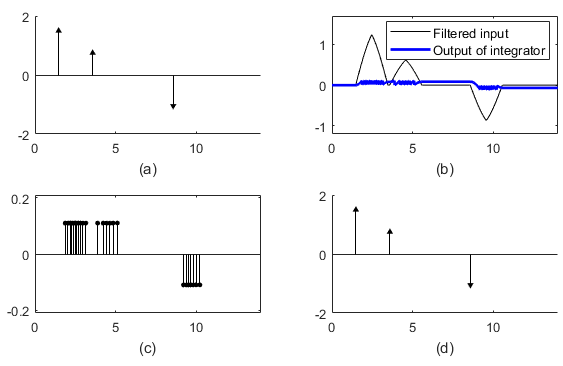}
\caption{Sampling of a stream of Diracs using the integrate-and-fire TEM. The input is shown in (a), the filtered input in (b), the output non-uniform samples in (c), and the reconstructed signal in (d).}
\label{fig:int_stream_diracs}
\vspace{-1.8em}
\end{figure}

\vspace{-0.5em}
\subsection{Estimation of a Stream of Pulses}
Let us now consider a stream of pulses of the form $x(t)*g(t)$, where $x(t)$ is defined in Eq. (\ref{eq:input stream Diracs}) and the support of $g(t)$ is $[-\epsilon, \epsilon]$. Filtering this signal with the second-order E-spline $\varphi(t)$ is equivalent to filtering the stream of Diracs $x(t)$ with the kernel $(\varphi*g)(t)$. 
As a case in point, let us consider the cosine-squared pulse $g(t)=\cos^2(t)$, and assume that $2\epsilon < \frac{L}{2}$, where $L$ is the length of the filter's support.  In addition, suppose we want to estimate the first pulse $(x_1*g)(t)$ in the stream $x(t)*g(t)$ and denote its timing information with $t_1, t_2,...,t_N$. The first three output samples can be computed as follows:
\vspace{-0.5em}
\small
\begin{equation*}
y(t_1) = \int_{\tau_1}^{t_1} f(t) dt=\langle x_1(t), (\varphi*g*q_{\theta_1})(t-\tau_1) \rangle,
\end{equation*}
\normalsize
\small
\begin{equation}
\label{eq:output_sample_2_for_1_pulse}
\vspace{-0.5em}
y(t_2) = \langle x_1(t), (\varphi*g*q_{\theta_2})(t-t_1) \rangle,
\end{equation}
\normalsize
\small
\begin{equation}
\label{eq:output_sample_3_for_1_pulse}
\vspace{-0.5em}
y(t_3) = \langle x_1(t), (\varphi*g*q_{\theta_3})(t-t_2) \rangle,
\end{equation}
\normalsize
where $x_1(t) = x_1 \delta(t-\tau_1)$ is the first Dirac in the stream $x(t)$ with $x_1>0$, $\theta_1=t_1-\tau_1$, $\theta_2=t_2-t_1$ and $\theta_3=t_3-t_2$.

Assuming $2\epsilon<\frac{L}{2}$ we can leverage the results in Eq. (\ref{eq: convo_pulses_1}) to show that in the interval $(\epsilon - \frac{L}{2}, - \epsilon)$, we get $(\varphi*g)(t) = A e^{\alpha_0 t} + Be^{\alpha_1 t}$, for some constants $A$ and $B$. Then, in the interval $(t_2-\frac{L}{2}+\epsilon, t_1-\epsilon)$, the function $(\varphi*g*q_{\theta_2})(t-t_1)$ can also be expressed as a linear combination of the exponentials $e^{\alpha_0 t}$ and $e^{\alpha_1 t}$. Similarly, $(\varphi*g*q_{\theta_3})(t-t_2)$ is a linear combination of the same exponentials in the interval $(t_3-\frac{L}{2}+\epsilon, t_2-\epsilon)$. 
As a result, in the knot-free interval $I=(t_3-\frac{L}{2}+\epsilon, t_1-\epsilon)$, we can perfectly reproduce two exponentials as in Eq. (\ref{eq:pulses_small_epsilon}), using the shifted kernels $(\varphi*g*q_{\theta_{n+1}})(t-t_n)$, for $n=1,2$. We can then compute two signal moments as in Eq. (\ref{eq:integrator_moments_single_dirac}), and retrieve the amplitude and location of the first Dirac $x_1 \delta(t-\tau_1)$ in the stream $x(t)$ using Prony's method. 

In order for these derivations to hold we need to ensure that $\tau_1 \in I$, or in other words that $t_1>\tau_1+ \epsilon$ and $t_3<\tau_1+\frac{L}{2}-\epsilon$. 
Since the filtered input corresponding to the first pulse satisfies $x_1 (\varphi*g)(\tau_1-t)>0$, for $t \in (\tau_1-\epsilon, \tau_1+L+\epsilon)$ and $x_1 (\varphi*g)(-t+\tau_1)=0$ otherwise, the condition $t_1>\tau_1+ \epsilon$ holds provided the trigger mark of the comparator satisfies:
\small
\begin{equation} \label{eq:CT_pulses_1}
C_T > \int_{\tau_1-\epsilon}^{\tau_1+\epsilon}  (\varphi*g)(\tau_1-t) dt =  \int_{-\epsilon}^{\epsilon}  (\varphi*g)(-t) dt.
\end{equation}
\normalsize

Using the same reasoning as in Section \ref{subsection:Estimation of a Stream of Diracs Integrator}, the condition $t_3<\tau_1+\frac{L}{2}-\epsilon$ holds provided:
\small
\begin{equation} \label{eq:CT_pulses_2}
C_T <  \frac{A_{min}}{4}  \int_{-\epsilon}^{\frac{L}{2}-\epsilon}  (\varphi*g)(-t) dt,
\end{equation}
\normalsize
where $A_{min}$ is the minimum amplitude of the Diracs in $x(t)$.

We also note that in order for Eq. (\ref{eq:CT_pulses_1}) and (\ref{eq:CT_pulses_2}) to be simultaneously satisfied, we need to impose additional constraints on $\epsilon$, such that:
\vspace{-0.5em}
\small
\begin{equation*}
\int_{-\epsilon}^{\epsilon}  (\varphi*g)(-t) dt <  \frac{1}{4} \int_{-\epsilon}^{\frac{L}{2}-\epsilon}  (\varphi*g)(-t) dt.
\end{equation*}
\normalsize
Finally, once the first pulse centered around $\tau_1$ has been estimated, and assuming a minimum separation between consecutive pulses of at least $L+2\epsilon$ (which is the length of the support of $(\varphi*g)(t)$), we can use subsequent samples after $\tau_1+L+2\epsilon$ to retrieve the next pulse $(x_2 *g)(t)$ in the input. 

The sampling and perfect retrieval of a stream of cosine-squared pulses are depicted in Fig.~\ref{fig:pulse}, for $C_T=0.8$. 
\begin{figure}[H]
\vspace{-0.5em}
\centering
\includegraphics[width=0.43\textwidth]{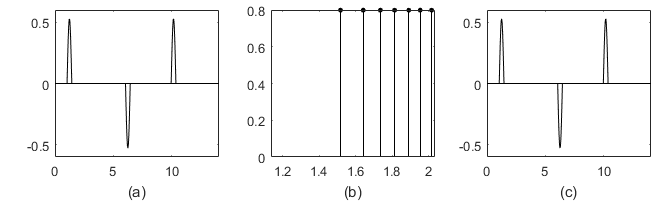}
\caption{Sampling of a stream of pulses using the integrate-and-fire TEM. The input is shown in (a), the non-uniform samples used for retrieval of the first pulse in (b), and the reconstructed signal in (c).}
\label{fig:pulse}
\end{figure}

\vspace{-1.5em}
\subsection{Multi-channel Estimation of Bursts of Diracs}
\label{subsection:Estimation of Bursts of Diracs with a Multi-channel Approach Integrator}
Let us now consider the estimation of a sequence of bursts of Diracs as in Eq. (\ref{eq:input burst Diracs}). This problem is equivalent to the estimation of a stream of Diracs, however, this time, the $K$ Diracs can be arbitrarily close to each other. Therefore, the estimation of a burst of $K$ Diracs involves retrieving a larger number of moments (at least $2K$) to accurately retrieve the Diracs. We employ a multi-channel scheme of $M \geq K$ different acquisition devices, each of which will help us compute 2 different signal moments. We will show that it is sufficient to record 2 output samples for each channel in order to perfectly reconstruct each burst in the input signal, and that the trigger mark of the threshold detector can be adjusted to ensure these output samples are located \textit{suitably} close to the input burst of Diracs. Specifically, we need to ensure that the 2 samples we use for estimation have contributions from all the $K$ Diracs, and hence, occur after the last Dirac in the burst.

\begin{prop}
\label{prop:Multi-channel Estimation of Bursts of Diracs Using the Integrate-and-fire TEM}
Let us consider a system of $M \geq K$ TEM devices as in Fig.~\ref{fig:integrate_and_fire_model_2}. The filter $\varphi(t)$ of the $m^{th}$ TEM is a second-order E-spline whose support has length $L$, and which can reproduce two different exponentials, $e^{j\omega_{m_0}t}$ and $e^{j\omega_{m_1}t}$, with $\omega_{m_0}=\omega_{0}+\lambda m$, $\lambda=\frac{-2\omega_0}{2M-1}$, $0<\omega_{0} \leq \frac{\pi}{L}$, $\omega_{m_1}=-\omega_{m_0}$, and $m=0,1,...,M-1$.
Moreover, let us assume the spacing between consecutive bursts is larger than $L$, and the maximum separation between the last and first Dirac in any burst $b$ satisfies $\tau_{b,K}-\tau_{b,1}< \frac{L}{2}$. In addition, suppose that the comparator's trigger mark $C_T$ satisfies the following conditions for each device $m$ and burst $b$:
\small
\begin{equation}
\vspace{-0.5em}
\label{eq:threshold_cond_summary_1}
C_T >\frac{(K-1)A_{\max}}{\omega_{m_0}^2} [1-\cos(\omega_{m_0}(\tau_{b,K}-\tau_{b,1}))] ,
\end{equation}
\begin{equation}
\vspace{-0.5em}
\label{eq:threshold_cond_summary_2}
C_T <   \frac{KA_{\min}}{5\omega_{m_0}^2} [1-\cos(\omega_{m_0}(\frac{L}{2}-(\tau_{b,K}-\tau_{b,1})))],
\end{equation}
\normalsize
where $A_{\max}$ and $A_{\min}$ are the absolute maximum and minimum amplitudes of the input, respectively.

Then, the timing information $t_{1,m},t_{2,m},...,t_{N,m}$ for $m=0,1,...,M-1$ provided by $M$ devices as in Fig.~\ref{fig:integrate_and_fire_model_2} is a sufficient representation of bursts of $K$ Diracs as in Eq. (\ref{eq:input burst Diracs}).

\end{prop}
\begin{proof}
See Appendix \ref{appendix:Multi-channel Estimation of Bursts of Diracs Using the Integrate-and-fire TEM}.
\end{proof}
\vspace{-1em}
Even though we considered the sampling of bursts of Diracs using a multi-channel system, it is possible under slightly more restrictive conditions, to achieve the same using a single TEM device. Therefore, for the sake of completeness, we state the following result without proof:

\begin{prop}
\label{prop:Single-channel Estimation of Bursts of Diracs}
Let us consider the integrate-and-fire TEM in Fig.~\ref{fig:integrate_and_fire_model_2}.
Let the sampling kernel $\varphi_P(t)$ be an E-spline of order $P\geq2K$ and  support of length $L$, which can reproduce $P$ different exponentials $e^{j\omega_mt}$, with $\omega_m=\omega_0+m\lambda$, $m=0,1,...,P-1$, and $0<\omega_0\leq \frac{\pi}{L}$. In addition, setting $P$ even and $\lambda=\frac{-\pi}{P}$ ensures $\varphi(t)$ is a real-valued function. In this setting, let us assume the spacing between bursts is larger than $L$, and the separation between the last and first Diracs within any burst $b$ satisfies $\tau_{b,K}-\tau_{b,1}<\frac{L}{P}$. In addition, suppose the trigger mark of the comparator $C_T$ satisfies:
\small
\begin{equation}
\label{eq:threshold_cond_single_ch_summary_1}
C_T > (K-1)A_{\max}\int_{0}^{\Delta_b} \varphi(-\tau) d\tau,
\end{equation}
\begin{equation}
\label{eq:threshold_cond_single_ch_summary_2}
C_T <   \frac{ K A_{\min}}{P+3} \int_{0}^{\frac{L}{P}} \varphi(-\tau) d\tau,
\vspace{-0.3em}
\end{equation}
\normalsize
where $\Delta_b=\max(\tau_{b,K}-\tau_{b,1})$.

Then, the timing information $\{t_1, t_2, ..., t_N\}$ provided by the integrate-and-fire TEM in Fig.~\ref{fig:integrate_and_fire_model_2} is a sufficient representation of a sequence of bursts of $K$ Diracs as in Eq. (\ref{eq:input burst Diracs}).
\vspace{-0.5em}
\end{prop}

\vspace{-0.8em}
\subsection{Estimation of Piecewise Constant Signals}
\label{subsection:Estimation of Piecewise Constant Signals}
Let us now consider a piecewise constant signal $x(t)$, and assume that we filter this with the derivative of an E-spline $\varphi(t)$ of order $P\geq2$, obtained using Eq. (\ref{eq:higher_order_e_spline}).
Filtering $x(t)$ with $\frac{d\varphi(t)}{dt}$ ensures that in a region where the input is constant, there are no output spikes, since $\frac{d\varphi(t)}{dt}$ has average value equal to zero. This leads to energy-efficient sampling of the piecewise constant signal, resulting in a small average number of output spikes.
In this setting, the filtered input is given by:
\small
\begin{equation*}
\begin{split}
f(t)= x(t)* \frac{d\varphi(t)}{dt} = \frac{dx(t)}{dt}* \varphi(t).
\end{split}
\end{equation*}
\normalsize

This shows that filtering a piecewise constant signal $x(t)$ with $\frac{d\varphi(t)}{dt}$ is equivalent to filtering the stream of Diracs corresponding to the discontinuities of the piecewise constant signal with the E-spline $\varphi(t)$. 
The discontinuities of $\frac{dx(t)}{dt}$ can be estimated from the output spikes, by extending the results of Proposition \ref{prop:Estimation of a Stream of Diracs Integrator} to the case of a $P$-order E-spline $\varphi_P(t)$, with $P\geq2$.
In this case, the E-spline $\varphi_P(t)$ of support of length $L$ can reproduce $P\geq2$ different complex exponentials $e^{j\omega_mt}$, with $\omega_m=\omega_0+\lambda m$. and $m=0,1,...,P-1$. Moreover, choosing $\lambda=\frac{-2\omega_0}{P-1}$ and $P$ even ensures the kernel $\varphi_P(t)$ is a real-valued function.
As before, the separation between consecutive Diracs must be larger than $L$ and the trigger mark of the comparator satisfies:
\small
\begin{equation}
\label{eq:stream_diracs_higherorderspline_threshold}
0<C_T <\frac{A_{\min}}{P+2}\int_{0}^{\frac{L}{P}} \varphi_P(-\tau) d\tau.
\end{equation}
\normalsize

Suppose we want to estimate the $k^{th}$ discontinuity in $ \frac{dx(t)}{dt}$, of amplitude $z_k$ and located at $\tau_k$, and let us denote the locations of the first output spikes after $\tau_k$ with $t_n, t_{n+1},...t_N$.
Then, using a similar proof as in Section \ref{subsection:Estimation of a Stream of Diracs Integrator}, we can show that the constraint in Eq. (\ref{eq:stream_diracs_higherorderspline_threshold}) guarantees that $\tau_k \in I=(t_{n+P}-\frac{L}{P}, t_n)$. 
Then, we can compute the following signal moments:
\small
\begin{equation*}
\begin{split}
s_m &= \sum_{i=1}^{P} c_{m,n}^Iy(t_{n+i}) \aeq z_k \sum_{i=1}^{P}c_{m,n}^I(\varphi_P*q_{\theta_{n+i}})(\tau_k-t_{n+i-1})\\
&\beq z_k e^{j\omega_m\tau_k}, \text{ for } m=0,1,...,P-1.
\end{split}
\end{equation*}
\normalsize
In these derivations, $(a)$ follows from Eq. (\ref{eq: output_samples_equivalent_filter}), and $(b)$ holds given $\tau_k \in (t_{n+P}-\frac{L}{P}, t_n)$, and the fact that none of the kernels $(\varphi_P*q_{\theta_{n+i}})(\tau_k-t_{n+i-1})$ have any discontinuities in $(t_{n+P}-\frac{L}{P}, t_n)$, for $i=1,2,...,P$.
As before, we can use Prony's method to estimate $z_k$ and $\tau_k$ from the signal moments $s_m$.
Finally, we can retrieve the piecewise constant signal $x(t)$ once we have estimated its discontinuities $\frac{dx(t)}{dt}$.

The sampling and reconstruction of a piecewise constant signal are depicted in Fig.~\ref{fig:int_piece_ct}. The filter is the derivative of the fourth-order E-spline, with support length $L=4$, as seen in Fig.~\ref{fig:int_piece_ct}(b), the separation between input discontinuities is larger than the length of the kernel's support as depicted in Fig.~\ref{fig:int_piece_ct}(a), and the comparator's trigger mark is $C_T=0.001$. The estimation of the input is exact to numerical precision.
\begin{figure}[htb]
\vspace{-0.5em}
\centering
\includegraphics[width=0.4\textwidth]{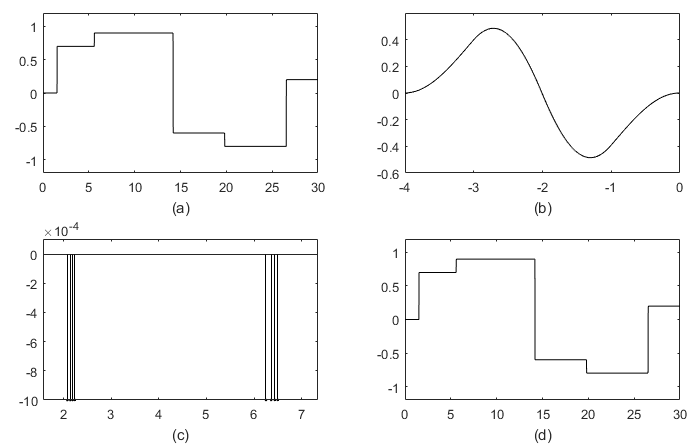}
\caption{Sampling of a piecewise constant signal with sufficiently separated discontinuities, using the integrate-and-fire TEM. The input is shown in (a), the sampling kernel in (b), the non-uniform samples used for estimation of the first two input discontinuities in (c), and the reconstructed signal in (d).}
\label{fig:int_piece_ct}
\vspace{-0.5em}
\end{figure}

Similarly, the results of Propositions \ref{prop:Multi-channel Estimation of Bursts of Diracs Using the Integrate-and-fire TEM} and \ref{prop:Single-channel Estimation of Bursts of Diracs} can be extended to the case of a piecewise constant signal $x(t)$, where the discontinuities $\frac{dx(t)}{dt}$ are bursts of arbitrarily close Diracs, as in Eq. (\ref{eq:input burst Diracs}). For example, in Fig. \ref{fig:piecewise_bursts}, we show the time encoding and perfect decoding of a piecewise constant signal, with two arbitrarily close discontinuities. 
\begin{figure}[htb]
\vspace{-0.5em}
\centering
\includegraphics[width=0.45\textwidth]{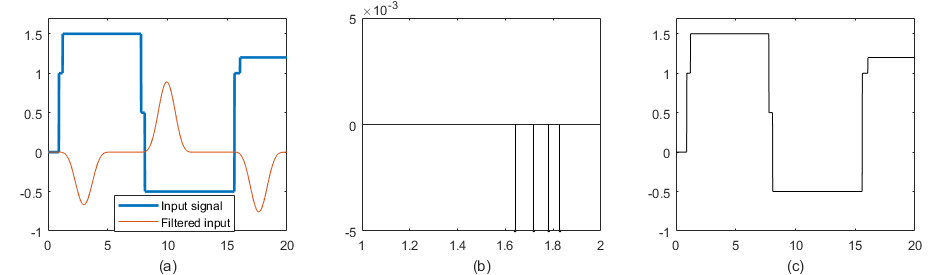}
\caption{Sampling of a piecewise constant signal, with arbitrarily close discontinuities, using the integrate-and-fire TEM. The input is shown in (a), the non-uniform samples used for estimation of the first burst of two discontinuities in (b), and the reconstructed signal in (c).}
\label{fig:piecewise_bursts}
\vspace{-0.5em}
\end{figure}

We conclude this section by summarizing possible choices of hyperparameters in our sampling framework based on an integrate-and-fire system. Specifically, let us consider the case of streams of bursts of $K$ Diracs, and discuss the relationship between the sampling kernel and the trigger mark of the comparator, and how these parameters determine the density of output samples. The sampling kernel is assumed to be the E-spline given in Eq. (\ref{eq:higher_order_e_spline}) of order $P$, and support length $L=P$.
Furthermore, the conditions of the trigger mark $C_T$ ensure that the output samples used for reconstruction are located \textit{sufficiently} close to the burst of Diracs, and in a region where the filtered input is continuous. As a result, these conditions depend on the separation $\Delta_b$ between the Diracs, as well as on the location of the knots of the sampling kernel, which in turn depends on the length of the support of this kernel. Setting $C_T$ to its maximum theoretical value ensures that the number of samples is minimised.

The choice of the hyperparameters of the integrate-and-fire TEM is summarised in Table \ref{table:table3}. Here, a burst of $K=2$ Diracs was time encoded using an $M$-channel system. The amplitudes of each of the Diracs was chosen uniformly at random in the interval $[1,2]$ and the trigger mark $C_T$ computed using Eq. (\ref{eq:threshold_cond_summary_2}). The results were averaged over 100 different experiments.
\vspace{-0.5em}
\small
\begin{table}[H]
\centering
\caption{Choice of hyperparameters for estimating a burst of $2$ Diracs, using an $M$-channel integrate-and-fire system.
 \label{table:table3}}
\setlength\extrarowheight{5pt} 
\begin{tabular}{|c|c|c|c|c|c|P{2.8cm}|} 
        \hline
      $M$ & $P$ & $L$& $C_T^{max}$ & $\omega_0$ & $\Delta_b$ &Average number of samples per burst\\
      \hline
       $2$ & $2$ & $2$& $0.1218$ & $-\frac{\pi}{3}$ & $0.2$ & 44.8\\
        \hline
       $2$ & $2$ & $2$& $0.0947$ & $-\frac{\pi}{3}$ & $0.3$ & 56.21\\
        \hline
       $1$ & $4$ & $4$& $0.0114$ & $-\frac{\pi}{3}$ & $0.2$ &202.24 \\
        \hline
       $2$ & $2$ & $2$& $0.1286$ & $-\frac{\pi}{2}$ & $0.2$ & 43.25\\
        \hline
       $2$ & $4$ & $4$& $0.0133$ & $-\frac{\pi}{3}$ & $0.2$ & 413.3 \\
        \hline
\end{tabular}

\end{table}
\normalsize
\vspace{-0.5em}
Finally, we make the remark that only some of the output samples are used for input reconstruction. For online reconstruction applications, these are the only samples that need to be stored. This is depicted in Fig.~\ref{fig:sequential_bursts}, where we highlight in red the samples used for reconstruction of each burst of Diracs, of one of the two channels. Only the second and third output samples, located at $t_2^{(1)}$ and $t_3^{(1)}$ need to be recorded and used to retrieve the first burst of 2 Diracs. Once the first burst has been estimated, we record the second and third output samples after $\tau_2+L$, located at $t_2^{(2)}$ and $t_3^{(2)}$ in order to retrieve the next burst. 

\begin{figure}[htb]
\vspace{-1em}
\centering
\includegraphics[width=0.37\textwidth]{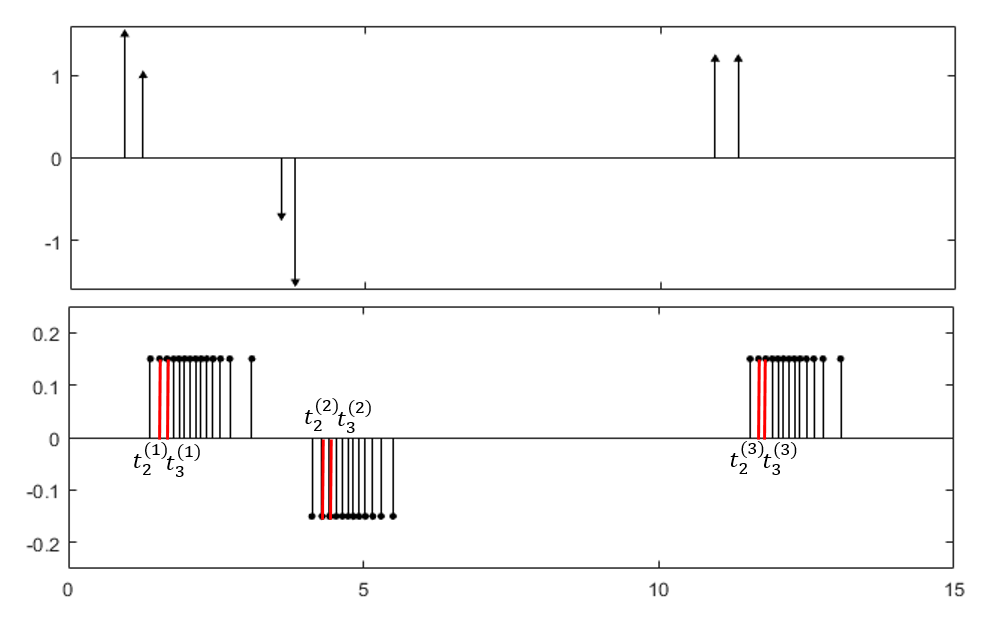}
\caption{Time encoding of a sequence of bursts of 2 Diracs, using an integrate-and-fire system. The input is shown in the top plot, and the output samples in the bottom plot.}
\label{fig:sequential_bursts}
\vspace{-1em}

\end{figure}
\vspace{-1em}

\section{Generalized Time-based Sampling}
\label{section:Generalized Time-based Sampling}
To highlight the potential practical implications of the methods developed in the previous sections, we present here extensions of our framework to deal with arbitrary kernels and the noisy scenario, and show that reliable input reconstruction can be achieved also in these scenarios.
\vspace{-0.6em}
\subsection{Sampling with Arbitrary Kernels}
\label{section:Universal Sampling Based on Timing}
In the previous sections we have presented methods for perfect retrieval of certain classes of non-bandlimited signals from timing information. We have seen that these methods require the sampling kernel $\varphi(t)$ to locally reproduce exponentials, in order to be able to map this problem to Prony's method.
In reality, however, the sampling kernel may not have the exponential reproducing property as in Eq. (\ref{eq:exp_spline non-uniform}).
Let us now consider an arbitrary kernel $\tilde{\varphi}(t)$, and find a linear combination of its non-uniform shifted versions that gives the best approximation of $P$ exponentials $f(t)=e^{j\omega_mt}$ within an interval $I$, for $\omega_m=\omega_0+\lambda m$, $m=0,1,...,P-1$, and $\lambda = \frac{-2\omega_0}{P-1}$. In other words, we want to find the optimal coefficients $c_{m,n}^I$ such that:
\small
\begin{equation}
\label{eq:approx sf recon exp}
\sum_{n=1}^N c_{m,n}^I \tilde{\varphi}(t-t_n) \approx e^{j\omega_mt},
\end{equation}
\normalsize
for $t \in I$ and $n=1,2,...,N$, with $N$ being the number of kernels $\tilde{\varphi}(t-t_n)$ overlapping $I$.

We find the coefficients $c_{m,n}$ using the \textit{least-squares approximation} method described in \cite{6578165}. The coefficients are computed using the orthogonal projection of $f(t)$ onto the space spanned by the non-uniform shifts $\tilde{\varphi}(t-t_n)$, such that:
\small
\begin{equation}
\label{eq:first eq approx sf}
\langle f(t)-\sum_{k=1}^N c_{m,k}^I \tilde{\varphi}(t-t_k),  \tilde{\varphi}(t-t_n)\rangle = 0,
\end{equation}
\normalsize
for $t \in I$ and $n=1,2,..N$.

Furthermore, Eq. (\ref{eq:first eq approx sf}) is equivalent to:
\small
\begin{equation*}
\langle f(t),  \tilde{\varphi}(t-t_n)\rangle = \sum_{k=1}^N c_{m,k}^I  \langle \tilde{\varphi}(t-t_k),\tilde{\varphi}(t-t_n)\rangle,
\end{equation*}
\normalsize
which represents a system of $N$ equations from which we can determine the $N$ coefficients $c_{m,k}^I $, for each $m=0,1,...,P-1$.

We then use the calculated coefficients $c_{m,k}^I$ to compute the signal moments as in Section \ref{section:Perfect Recovery of Signals from Timing Information obtained with an Integrate-and-fire System}. 
Finally, the estimation of the input can be further refined using the Cadzow iterative algorithm in order to increase the accuracy of the signal moments, before applying Prony's method \cite{4472241, 1488}.

The sampling and reconstruction of bursts of 2 Diracs are depicted in Fig.~\ref{fig:p3_signals_exp4}. We use the multi-channel estimation method presented in Section \ref{subsection:Estimation of Bursts of Diracs with a Multi-channel Approach Integrator}, where the filter of each channel is a third order B-spline $\beta_3(t)$, such that the modified kernel $(\beta_3*q_{\theta_n})(t)$ in Eq. (\ref{eq:modified samples integrator summary}) cannot reproduce exponentials. 
Moreover, we aim to approximately reproduce 4 different exponentials for each channel, and hence we require a number of 4 non-uniform samples, as discussed in Section \ref{subsection:Sampling Kernels}. In Fig.~\ref{fig:p3_recon_exp}, we depict the approximate exponential reproduction in Eq.~(\ref{eq:approx sf recon exp}), within the interval $I=(0.82, 1.4)$ overlapping the first burst of Diracs. The mean-squared error of the exponential reproduction within this interval is $9.47 \times 10^{-13}$.
Finally, the estimation of the input is close to exact, as depicted in Fig.~\ref{fig:p3_signals_exp4}(c). In particular, the mean-squared error in the time locations of the Diracs is $2.44\times 10^{-4}$, and the mean-squared error in the amplitudes of the Diracs is $2.04\times 10^{-10}$.
\vspace{-0.5em}
\begin{figure}[htb]
\vspace{-0.5em}
\centering
\includegraphics[width=0.23\textwidth]{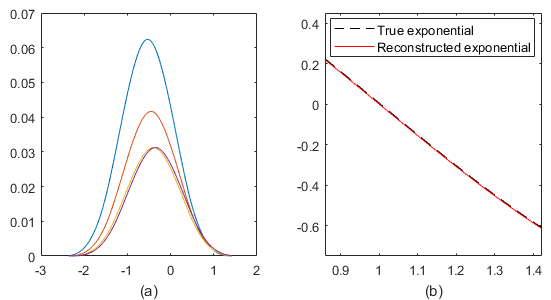}
\caption{Approximate exponential reproduction using non-uniform shifts of the kernel $(\beta_3*q_{\theta_n})(t)$. The kernels are shown in (a), and the exponential reproduction using these shifted kernels in (b).}
\label{fig:p3_recon_exp}
\vspace{-1em}
\end{figure}

\begin{figure}[htb]
\centering
\vspace{-0.5em}
\includegraphics[width=0.45\textwidth]{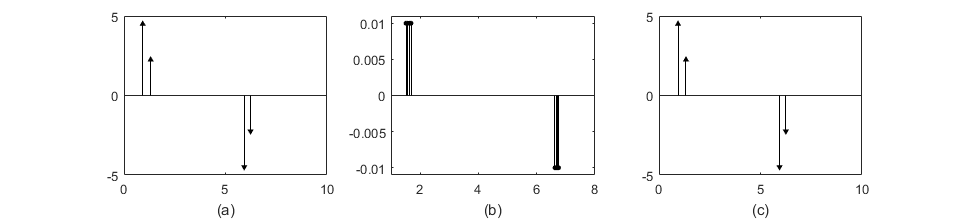}
\caption{Universal sampling of a sequence of bursts of Diracs using the integrate-and-fire TEM. The input signal is shown in (a), the output non-uniform samples of one channel used for estimation in (b), and the reconstructed signal in (c).}
\label{fig:p3_signals_exp4}
\vspace{-1em}
\end{figure}

\vspace{-0.5em}
\subsection{Robustness of the Integrate-and-fire TEM to Noise}
\label{section:Robustness to Noise}
In many practical circumstances, the input signal is corrupted by noise, which is typically assumed to be white, additive Gaussian noise.
When this happens, the non-uniform times $\{t_n\}$ change which means that the sequence of moments $s_m$ is also corrupted, and perfect reconstruction may no longer be possible.
Nevertheless, if the noise has average value equal to 0, it is in part removed by the integrator in the TEM, as a result of the averaging effect of the integral.

In Fig.~\ref{fig:noise} we show the reconstruction of a piecewise constant signal corrupted by white, additive Gaussian noise, using the method in Section \ref{subsection:Estimation of Piecewise Constant Signals}. The filter is the derivative of a fourth-order E-spline with support length $L=4$ which can reproduce the exponentials $e^{\pm j \frac{\pi}{3}t}$ and $e^{\pm j \frac{\pi}{6}t}$, the trigger mark of the comparator is $C_T=0.001$, the standard deviation of the noise is $\sigma=0.1$ (SNR$=21.56$dB), and the separation between consecutive discontinuities of the input is larger than $L$. The reconstruction of the input from noisy samples is very accurate.
A quantitative analysis of the effect of noise on the retrieval of this piecewise constant signal is presented in Table \ref{table:table1}. The table shows the error of the estimated locations and the relative error of the estimated amplitudes of the discontinuities in the input signal, averaged over 10000 experiments.
\vspace{-0.8em}
\begin{figure}[htb]
\centering
\includegraphics[width=0.4\textwidth]{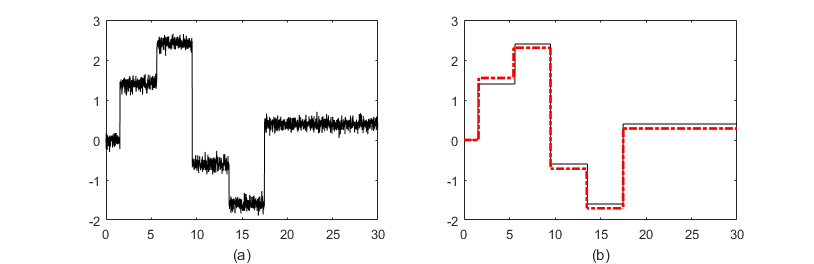}
\caption{Estimation of a piecewise constant signal from noisy samples, obtained using the integrate-and-fire TEM. The noisy input is shown in $(a)$, and the reconstruction in $(b)$.}
\label{fig:noise}
\vspace{-1em}
\end{figure}
\vspace{-0.2em}
\small
\begin{table}[ht]
\begin{center}
\caption{Effect of noise on the estimation of a piecewise constant signal, from spikes obtained using the integrate-and-fire TEM. The error $\epsilon_t$ is the average absolute difference between the true and estimated locations, $\epsilon_A$ is the relative error of the estimated amplitudes of the input discontinuities and SER is the signal-to-error-ratio, for amplitude estimation.
 \label{table:table1}}
\begin{tabular}{|c|c|c|c|c|} 
        \hline
     \text{SNR(dB)} & \textbf{$\sigma$} & \textbf{$\epsilon_t$} & \textbf{$\epsilon_A$}& \text{SER(dB)}\\
      \hline
      $43.33$ &0.01 & $2.61 \times 10^{-4}$ &$6.21 \times 10^{-5}$ & $43.11$ \\
        \hline
      $29.63$&0.05 &$0.0015$ &$2.1509 \times 10^{-4}$ & $24.59$ \\
        \hline
     $23.38$ & 0.1 &$0.0042$ &$0.0026$ & $22.83$ \\
        \hline
\end{tabular}
\end{center}
\vspace{-1em}
\end{table}
\normalsize

In Fig.~\ref{fig:integrator_snr} we show the reconstruction errors, for the case of a stream of Diracs, for different SNR values, averaged over 1000 experiments. Here, the input signal is corrupted by white, additive Gaussian noise, and the sampling kernel is a second-order E-spline whose support has length $L=2$, defined as in Eq. (\ref{eq: first_order_e_spline_definition_1}), for $\alpha_0 = -\alpha_1 = j \frac{\pi}{3}$. When SNR$=20dB$, the amplitude mean-squared error is $1.17 \times 10^{-3}$ and the mean-squared error in time locations is $2.5\times 10^{-3}$.
Finally, in Fig.~\ref{fig:integrator_bursts_noisy} we depict the estimation of an input stream of Diracs corrupted by noise, from its timing information, when SNR$ = 10dB$.

\vspace{-0.8em}
\begin{figure}[htb]
\centering
\includegraphics[width=0.45\textwidth]{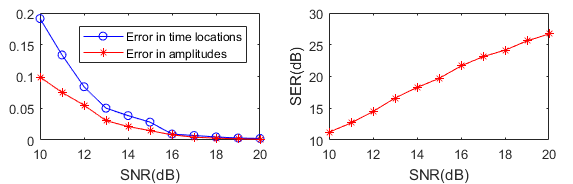}
\caption{Left: Average mean-squared errors in estimated time locations and amplitudes of a stream of Diracs corrupted by white, additive Gaussian noise; Right: Average signal-to-error ratio (SER) along signal-to-noise ratio, for amplitude estimation.}
\label{fig:integrator_snr}
\vspace{-1em}
\end{figure}
\vspace{-0.2em}

\vspace{-0.8em}
\begin{figure}[htb]
\centering
\includegraphics[width=0.45\textwidth]{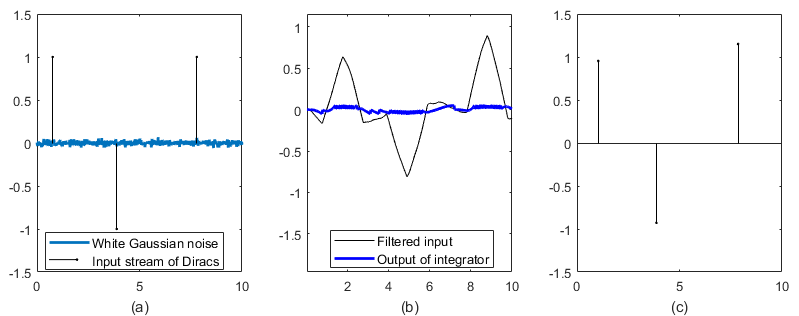}
\caption{Estimation of a stream of Diracs from noisy samples, obtained using the integrate-and-fire TEM. For SNR$=10dB$, the noisy input is shown in $(a)$, the filtered input and output of integrator in $(b)$, and the reconstruction in $(c)$.}
\label{fig:integrator_bursts_noisy}
\vspace{-1em}
\end{figure}

\vspace{-1em}
\subsection{Time Encoding and Decoding of Bursts of Diracs of Arbitrary Signs}
{In Section \ref{subsection:Estimation of Bursts of Diracs with a Multi-channel Approach Integrator} we presented sufficient conditions for perfect retrieval of bursts of Diracs defined as in Eq. (\ref{eq:input burst Diracs}). These conditions rely on various assumptions, including that the amplitudes $x_{b,k}$ of the Diracs $x_{b,k} \delta(t-\tau_{b,k})$ in the same burst $b$, have the same sign. In reality, this assumption may not always hold, and in this section we empirically show that the reconstruction framework presented in this paper usually performs well also when the amplitudes of the Diracs in the same burst have opposite signs.
We consider the estimation of a burst of 2 Diracs from its time encoded information using a 2-channel approach. We assume that the filter of each channel is a second-order E-spline defined as in Eq. (\ref{eq: first_order_e_spline_definition_1}) with support length $L=2$. We denote the output information of channel 1 with $t_{1,1}, t_{2,1},...,t_{N,1}$, and that of channel 2 with $t_{1,2}, t_{2,2},...,t_{N,2}$.
We assume that the amplitudes of these Diracs are distributed as Gaussian variables, of mean $\mu=0$ and variance $\sigma=1$.

The decoding scheme presented in Section \ref{subsection:Estimation of Bursts of Diracs with a Multi-channel Approach Integrator} showed that we can reliably use the samples $y(t_{3,1}), y(t_{4,1})$ of the first channel and $y(t_{3,2}), y(t_{4,2})$ of the second channel, in order to perfectly retrieve an input burst of 2 Diracs of same sign. The sufficient conditions on the trigger mark of the integrator given in  Eq. (\ref{eq:threshold_cond_summary_1}) and  (\ref{eq:threshold_cond_summary_2}) ensure that these samples are located after the second Dirac at $\tau_2$.
Here, we choose $C_T$ below its minimum theoretical value given in Eq. (\ref{eq:threshold_cond_summary_1}), in order to ensure a sufficient number of output samples is obtained, even when the two Diracs in the input have opposite signs and are located closely to each other. However, when lowering $C_T$, the samples $y(t_{3,1}), y(t_{4,1})$ and $y(t_{3,2}), y(t_{4,2})$ are not guaranteed to occur after $\tau_2$, and hence, may not be reliably used for reconstruction of both Diracs.
Therefore, we adjust the reconstruction scheme as follows.
Using $y(t_{2,1}), y(t_{3,1})$ of the first channel and $y(t_{2,2}), y(t_{3,2})$ of the second channel we compute the signal moments $s_m$ as described in Appendix \ref{appendix:Multi-channel Estimation of Bursts of Diracs Using the Integrate-and-fire TEM} and then build matrix $S$ as in Appendix \ref{appendix:Prony's method}. 
If the rank of the matrix $S$ is 1, then $t_{3,1}<\tau_2$ and $t_{3,2}<\tau_2$. Hence, we can use $y(t_{2,1}), y(t_{3,1})$ to estimate the first Dirac $x_1\delta(t-\tau_1)$. Once the first Dirac has been estimated, we remove its contribution from the output spikes, and use the next non-zero samples in order to estimate the second Dirac $x_2\delta(t-\tau_2)$.
Otherwise, if the rank of matrix $S$ is 2, then at least for one of the channels $i$, we get $t_{i,3}>\tau_2$. As a result of the similarity between the sampling kernels of the two channels, it is likely that $t_{i,3}>\tau_2$ for both $i=1$ and $i=2$. In other words, the samples $y(t_{4,1}), y(t_{5,1})$ and $y(t_{4,2}), y(t_{5,2})$ are likely to have contributions from both Diracs and hence, we can use the method in Section \ref{subsection:Estimation of Bursts of Diracs with a Multi-channel Approach Integrator} to estimate the burst.
In Table \ref{table:table2} we show the probability of correct estimation of the 2 Diracs, against different values of $\Delta_b$ and trigger mark $C_T$, averaged over 1000 experiments. The results show that we still achieve perfect reconstruction in most cases. The reconstruction typically fails when the number of samples required for estimation is not achieved (for example, when the amplitudes of the Diracs  are very small or when they have similar magnitudes, but opposite signs).

\vspace{-1em}
\small
\begin{table}[ht]
\begin{center}
\caption{Probability of perfect reconstruction of a burst of 2 Diracs, with random Gaussian amplitudes.
 \label{table:table2}}
\begin{tabular}{|c|c|c|} 
        \hline
      \textbf{$\Delta_b=\tau_2-\tau_1$} & \textbf{$C_T$} & P\{perfect estimation\}\\
      \hline
      $0.05$ & $0.01$ &$0.950$  \\
        \hline
      $0.1$ & $0.01$ & $0.959$  \\
        \hline
      $0.5$ & $0.001$ & $0.959$  \\
        \hline
      $1.5$ & $0.001$ & $0.978$  \\
        \hline
\end{tabular}
\end{center}
\vspace{-1em}
\end{table}
\normalsize
\vspace{-1.5em}
\renewcommand{\thesection}{\Roman{section}}
\section{Density of Non-uniform Samples Obtained with an Integrate-and-fire TEM}
\label{section:Density of Non-uniform Samples}
In the previous sections, we have presented techniques for estimation of non-bandlimited signals from timing information. We have seen that perfect estimation can be achieved using simple algorithms, and physically realisable kernels.
In this section we outline the fact that in many settings sampling based on timing using our integrate-and-fire system is an efficient way to acquire signals, resulting in a smaller density of samples, compared to classical sampling.

As a case in point we consider the retrieval of bursts of $K$ Diracs, described in Section \ref{subsection:Estimation of Bursts of Diracs with a Multi-channel Approach Integrator}. We have seen that perfect reconstruction from timing information can be achieved, provided the separation between consecutive bursts is at least $L$, and that the Diracs within any burst are sufficiently close. In particular, let us denote the maximum separation between the last and first Dirac within a burst with $\Delta=\max(\tau_K-\tau_1)<\frac{L}{2}$, which can be determined according to Eq. (\ref{eq:threshold_cond_summary_1}) and (\ref{eq:threshold_cond_summary_2}). Moreover, let us assume the input is sufficiently sparse, such that the average separation between consecutive bursts is $L+S$, with $S>0$.
Under these assumptions, the results in \cite{4156380} show that in order to retrieve the $K$ Diracs from uniform samples, we need at least $2K$ samples within the interval $L-\Delta$ following the burst of Diracs. As a result, the uniform sampling period must satisfy $T \leq \frac{L-\Delta}{2K}$.
Then, the number of uniform samples we record within an interval of length $L+S$ is $\frac{L+S}{T} =  \frac{2K(L+S)}{L-\Delta}$.
On the other hand, in the case of time encoding using the integrate-and-fire TEM in Fig.~\ref{fig:integrate_and_fire_model_2}, the results in Section \ref{subsection:Estimation of Bursts of Diracs with a Multi-channel Approach Integrator} show that we need to record 4 output samples for each of the $K$ channels (or equivalently, $4K$ samples for the case of single-channel sampling), for each burst of $K$ Diracs. We note that Eq. (\ref{eq:threshold_cond_summary_2}) shows that in many situations, the TEM outputs more than 4 spikes per channel. Nevertheless, these samples can be discarded since they are not used in estimation. For example, one way to stop recording spikes once we have obtained 4 non-zero samples, is to increase the trigger mark $C_T$ of the comparator in Fig.~\ref{fig:integrate_and_fire_model_2}, for a duration of $L-\Delta$.

Moreover, when the input is constant (zero), the integrate-and-fire TEM does not fire, and hence there are no output samples. Therefore, in an interval of size $L+S$, the number of stored samples from a $K$-Dirac burst is $4K$, $\forall S$. 

Furthermore, $\frac{2K(L+S)}{L-\Delta}>4K$ for $S\geq L-2\Delta>0$ and $\forall K$, which shows that the average number of non-uniform spikes required for the retrieval of $K$ Diracs is lower than the number of uniform samples required to estimate the same number of free input parameters, when the input is sufficiently sparse.

\vspace{-0.5em}
\section{Conclusions}
\label{sec:Conclusions}
In this work we established time encoding as an alternative sampling method for some classes of signals that are neither bandlimited, nor belong to shift-invariant subspaces. The proposed sampling scheme is based on first filtering the input signal, before retrieving the timing information using a crossing or integrate-and-fire TEM. 
We demonstrated sufficient conditions for the exact recovery of streams of Diracs, streams of pulses and piecewise constant signals, from their time-based samples. Central to our reconstruction methods is the use of specific filters that we proved can locally reproduce polynomials or exponentials. We further highlighted the potential of this new framework by showing that it is resilient to noise and that it can handle non-ideal filters. Finally, the diverse applications of previous results of finite rate of innovation theory \cite{7857059, 7736135, 5686950, 1329542} also serve as evidence for the potential for real-world applications of the theoretical framework developed in this paper. 

\vspace{-1em}
\appendices
\section{Prony's method}
\label{appendix:Prony's method}
One way to solve the problem of estimating the parameters $\{b_k, u_k\}_{k=1}^K$ from the sequence $s_m=\sum_{k=1}^K b_k u_k^m$ is given by the annihilating filter method, also referred to as Prony's method \cite{Prony}. The name of this approach comes from the observation that if we filter $s_m$ with a filter which has zeros at $\{u_k\}_{k=1}^K$, the output is zero, or in other words, this filter annihilates the sequence $s_m$.

The z-transform of the annihilating filter satisfies:
\small
\begin{equation}
\label{eq:annihilating filter}
H(z) = \sum_{m=0}^K h_m z^{-m} = \prod_{k=1}^K(1-u_k z^{-1}),
\end{equation}
\normalsize
which evaluates to zero when $z=u_k$.

Filtering the sequence $s_m$ with $h_m$ corresponds to the convolution of these sequences:
\small
\begin{equation}
\label{eq:annihilating filter coeffs}
\begin{split}
h_m*s_m = \sum_{l=0}^K h_l s_{m-l} =\sum_{k=1}^K b_k u_k^{m}\sum_{l=0}^K h_lu_k^{-l} \aeq 0,
\end{split}
\end{equation}
\normalsize
where $(a)$ holds since $z=u_k$ gives $H(z)=0$ in Eq. (\ref{eq:annihilating filter}).

Eq. (\ref{eq:annihilating filter coeffs}) can be written in matricial form as follows:
\small
\begin{equation}
\label{annihilating filter coeffs matrix}
\begin{bmatrix}
    s_K  & s_{K-1}& \cdots & s_0 \\
    s_{K+1}  & s_K& \cdots &  s_1 \\
   \vdots & \vdots & \ddots & \vdots \\
   s_{2K-1}   & s_{2K-2} & \cdots & s_{K-1} 
\end{bmatrix}
\begin{bmatrix}
  1     \\
h_1      \\
\vdots \\
h_K
\end{bmatrix}
=\textbf{Sh} = 0.
\end{equation}
\normalsize

It can be shown that provided $\{b_k\}_{k=1}^K$ are non-zero and $\{u_k\}_{k=1}^K$ are distinct, matrix \textbf{S} has full row rank $K$, which means the solution \textbf{h} given by Eq. (\ref{annihilating filter coeffs matrix}) is unique. Moreover, the solution \textbf{h} can be obtained by performing a singular value decomposition of \textbf{S}, where \textbf{h} is the singular vector corresponding to the zero singular value. 

Then, once the coefficients $h_m$ of the polynomial $H(z)$ are known, the parameters $\{u_k\}_{k=1}^K$ are obtained from the roots of this filter.
Finally, once $\{u_k\}_{k=1}^K$ are found, the parameters $\{b_k\}_{k=1}^K$ can be computed from the linear system of $K$ equations given by $s_m=\sum_{k=1}^K b_k u_k^m$, with $m=0,1,...,K-1$.

\vspace{-1em}
\section{}
\label{Appendix: Comparator Bursts Diracs}
\vspace{-0.5em}
\subsection{Proof of Proposition \ref{prop: Comparator Bursts Diracs}}
For simplicity, let us assume the number of devices equals the number of Diracs in a burst, i.e. $M=K$.
Suppose we want to estimate the Diracs in the first burst, located at $\tau_{1,1}, ....,\tau_{1,K}$. Moreover, assume for simplicity that their amplitudes satisfy $x_{1,1},...,x_{1,K}>0$.
In addition, let us consider the output of the $m^{th}$ TEM device, and denote its timing information with $\{t_1, t_2,...,t_N\}$.

Since we assume all the amplitudes in the first burst satisfy $0<x_{1,k}<A_{max}$, and since $0 \leq \varphi(t)<1$, we get $0\leq y(t)$ and $y(t)=\sum_{k=1}^K x_k \varphi(\tau_k-t)< KA_{max}< A=\max(g(t))$. 

Then, Bolzano's intermediate value theorem \cite{Bolzano} guarantees that the $m^{th}$ TEM outputs at most one sample in the interval $(\tau_{1,1}, \tau_{1,K})$, given the assumption $\tau_{1,K}-\tau_{1,1}< \frac{T_s}{2}$, and the fact that $0\leq y(t)<\max(g(t))$.
At the same time, this theorem also guarantees that the filtered input $y(t)$ crosses the sinusoidal reference signal in at least 3 points, within the window $(\tau_{1,1}, \tau_{1,1}+\frac{7T_s}{4})$, such that $t_3-\tau_{1,1}\leq \frac{7T_s}{4}$. Moreover, the assumption $T_s\leq\frac{2L}{7}$ ensures that $t_3-\tau_{1,1}\leq \frac{L}{2}$.
Hence, whilst the spike at $t_1$ may occur before $\tau_{1,K}$, the second and third spikes satisfy $t_{2},t_{3}\in (\tau_{1,K}, \tau_{1,1}+\frac{L}{2})$, which means that $\tau_{1,1}, \tau_{1,2},...\tau_{1,K} \in (t_{3}-\frac{L}{2},t_{2})$. 

Since in the interval $I=(t_{3}-\frac{L}{2},t_{2})$ there are no knots of either $\varphi(t-t_{2})$ or $\varphi(t-t_{3})$, we can compute the following signal moments for the $m^{th}$ channel:
\small
\begin{equation*}
\begin{split}
s_{m_i}&= \sum_{n=2}^3 c_{m_i,n}^{I} y(t_{n}) \aeq  \sum_{n=2}^{3} c_{m_i,n}^{I} \langle x(t), \varphi(t-t_{n})\rangle\\
&\beq \int_{-\infty}^{\infty} x(t) \sum_{n=2}^{3} c_{m_i,n}^{I} \varphi(t-t_{n}) dt \ceq \int_{-\infty}^{\infty} x(t) e^{j \omega_{m_i} t} dt \\
&\deq \int_{I} \sum_{k=1}^K x_{1,k} \delta(t-\tau_{1,k})  e^{j \omega_{m_i} t}dt =  \sum_{k=1}^K x_{1,k} e^{j \omega_{m_i} \tau_{1,k}}.
\end{split}
\end{equation*}
\normalsize
where $i\in\{0,1\}$, and $\omega_{m_0} = \omega_0 + \lambda m$ and $\omega_{m_1}=-\omega_{m_0}$.

In the derivations above, $(a)$ follows from Eq. (\ref{eq:non-uniform samples comparator}), $(b)$ from the linearity of the inner product, and $(c)$ from the local exponential reproduction property of the sampling kernel described in Eq. (\ref{eq:exp_spline non-uniform}), for $N=2$. Moreover, $(d)$ follows from Eq. (\ref{eq:input burst Diracs}), and given that $\tau_{1,1}, \tau_{1,2},...\tau_{1,K} \in (t_{3}-\frac{L}{2},t_{1})$.

By using the same approach on each of the $K$ channels, we can retrieve $2K$ different moments and, due to the specific choice of exponents, the $2K$ moments can be expressed as:
\vspace{-0.5em}
\small
\begin{equation*}
s_p = \sum_{k=1}^K x_{1,k}  e^{j \omega_0 \tau_{1,k}} e^{j\lambda p\tau_{1,k}}=\sum_{k=1}^K b_k u_k^p,
\end{equation*}
\normalsize
where $b_k= e^{j \omega_0 \tau_{1,k}}$, $u_k= e^{j\lambda \tau_{1,k}}$, and $p=0,1,...,2M-1$.

We can then apply Prony's method on $s_p$ to retrieve the $K$ amplitudes and the $K$ locations of the Diracs.
Finally, we use subsequent output samples, located after $\tau_{1,K}+L$ to retrieve the free parameters of the Diracs in the second burst, and we reiterate the process for the following bursts.

The sampling and reconstruction of a sequence of bursts of 2 Diracs are depicted in Fig.~\ref{fig:comp_bursts_2_diracs}. Here, the sampling kernel is a second-order E-spline for each channel, of support of length $L=2$, shown in Fig.~\ref{fig:comp_bursts_2_diracs}(c) and (d). The first channel's kernel reproduces the exponentials $e^{\pm j \frac{ \pi}{3}t}$, whereas the second kernel reproduces $e^{\pm j \frac{\pi}{9}t}$. Moreover, the comparator's reference signal has frequency $f_s=1.76>\frac{7}{2L}$, and the separation between consecutive bursts of Diracs is at least $L$. The amplitudes and locations of the estimated Diracs are exact.

\vspace{-0.9em}
\begin{figure}[htb]
\centering
\includegraphics[width=0.5\textwidth]{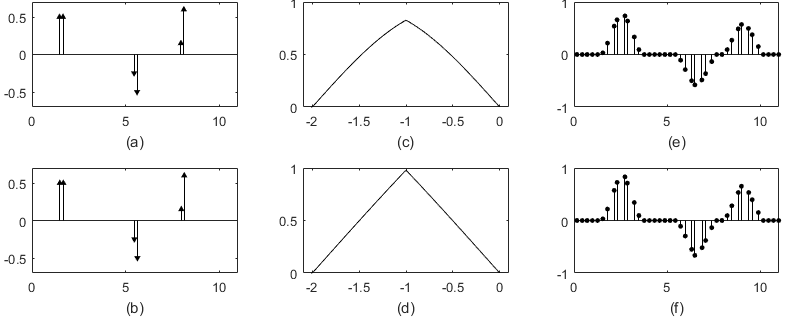}
\vspace{-0.5em}
\caption{Sampling of bursts of Diracs using the crossing TEM. The input signal is shown in (a), the reconstructed signal in (b), the sampling kernels of both channels in (c) and (d) respectively, and the corresponding non-uniform samples in (e) and (f).}
\label{fig:comp_bursts_2_diracs}
\vspace{-1em}
\end{figure}

\vspace{-1.2em}
\section{}
\label{appendix:Multi-channel Estimation of Bursts of Diracs Using the Integrate-and-fire TEM}
\subsection{Proof of Proposition \ref{prop:Multi-channel Estimation of Bursts of Diracs Using the Integrate-and-fire TEM}}
The input stream of bursts of Diracs can be sequentially estimated as follows. We estimate the first burst using the first four non-zero samples of each channel and the methods presented below. We then retrieve the second burst using the first four non-zero samples of each channel located after $\tau_{1,K}+L$, where $\tau_{1,K}$ denotes the estimated location of the last Dirac in the first burst, and $L$ is the length of the kernel's support. We then use the first non-zero samples located after $\tau_{2,K}+L$ to estimate the third burst, and repeat this procedure to estimate the subsequent bursts of Diracs.

Let us assume we want to retrieve burst $b$ and denote with $t_n,t_{n+1},t_{n+2}, t_{n+3}$ the first four output spikes located after $\tau_{b-1,K}+L$.
Then we have that $t_n>\tau_{b,1}>t_{n-1}$, where $\tau_{b,1}$ is the location of the first Dirac in the $b^{th}$ burst.
Furthermore, let us assume for simplicity that the Diracs in the $b^{th}$ burst satisfy $x_{b,1},...,x_{b,K}>0$, as depicted in Fig.~\ref{fig:int_bursts_all_signals}. 

\begin{figure}
\centering
\includegraphics[width=0.5\textwidth]{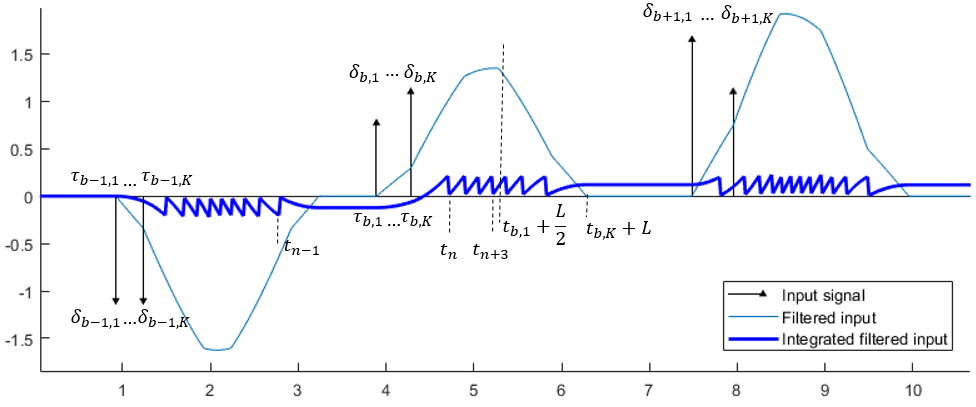}
\caption{Time encoding of a sequence of 2 bursts of 2 Diracs, when the amplitudes of the Diracs in a burst have the same sign.}
\label{fig:int_bursts_all_signals}
\vspace{-1.8em}
\end{figure}

In what follows, we show that the samples $y(t_{n+2})$ and $y(t_{n+3})$ can be reliably used to estimate the $b^{th}$ burst. 

We first prove that the following conditions hold:
\small
\begin{equation}
\vspace{-0.8em}
\label{eq:burst_integrator_sample_cond_1}
t_{n+1}>\tau_{b,K},
\end{equation}
\normalsize
and:
\vspace{-0.5em}
\small
\begin{equation}
\vspace{-0.7em}
\label{eq:burst_integrator_sample_cond_2}
t_{n+3}<\tau_{b,1}+\frac{L}{2}.
\end{equation}
\normalsize

We note that since we assume $x_{b,1},...,x_{b,K}>0$, the filtered input defined in Eq. (\ref{eq:filtered input integrator}) satisfies $f(\tau)>0$, and hence the condition in Eq. (\ref{eq:burst_integrator_sample_cond_1}) is equivalent to:
\small
\begin{equation}
\vspace{-0.5em}
\label{eq:burst_integrator_inequality_1}
\int_{\tau_{b,1}}^{t_{n+1}} f(\tau)d\tau>\int_{\tau_{b,1}}^{\tau_{b,K}} f(\tau)d\tau.
\end{equation}
\normalsize

The left-hand side of this inequality can be expressed as:
\small
\begin{equation}
\vspace{-0.5em}
\label{eq:burst_integrator_inequality_2}
\begin{split}
\int_{\tau_{b,1}}^{t_{n+1}} f(\tau)d\tau &= \int_{t_{n-1}}^{t_{n+1}} f(\tau)d\tau -\int_{t_{n-1}}^{\tau_{b,1}} f(\tau)d\tau \\
&\aeq 2C_T-\int_{t_{n-1}}^{\tau_{b,1}} f(\tau)d\tau \bg C_T,
\end{split}
\end{equation}
\normalsize
where $(a)$ holds given Eq. (\ref{eq:non_uniform_samples integrator}) and $(b)$ since $t_n>\tau_{b,1}>t_{n-1}$.

The right-hand side of Eq. (\ref{eq:burst_integrator_inequality_1}) can be re-written as:
\small
\begin{equation*}
\vspace{-0.5em}
\begin{split}
 \int_{\tau_{b,1}}^{\tau_{b,K}} f(\tau)d\tau &\cl \sum_{k=1}^{K-1} A_{\max}\int_{\tau_{b,k}}^{\tau_{b,K}} \varphi(\tau_{b,k}-\tau)d\tau  \\
&\dl \frac{(K-1)A_{\max}}{\omega_{m_0}^2} [1-\cos(\omega_{m_0}(\tau_{b,K}-\tau_{b,1}))] \\
&\el C_T \fl \int_{\tau_{b,1}}^{t_{n+1}} f(\tau)d\tau,
\end{split}
\end{equation*}
\normalsize
which proves the inequality in Eq. (\ref{eq:burst_integrator_sample_cond_1}).
 
In the derivations above, $(c)$ follows from the definition in Eq. (\ref{eq:filtered input integrator}) and since we assume $A_{max}>x_{b,1},...,x_{b,K}>0$. In addition, $(e)$ follows from Eq. (\ref{eq:threshold_cond_summary_1}) and $(f)$ from Eq. (\ref{eq:burst_integrator_inequality_2}). Finally, condition $(d)$ follows from:
\small
\begin{equation*}
\vspace{-0.5em}
\begin{split}
\int_{\tau_{b,k}}^{\tau_{b,K}}\varphi(\tau_{b,k}-\tau)d\tau &\hheq \frac{1}{\omega_{m_0}^2} [1-\cos(\omega_{m_0}(\tau_{b,K}-\tau_{b,k}))] \\
&\il  \frac{1}{\omega_{m_0}^2} [1-\cos(\omega_{m_0}(\tau_{b,K}-\tau_{b,1}))].
\end{split}
\end{equation*}
\normalsize
where $(h)$ follows from the definition of $\varphi(\tau_{b,k}-\tau)$ in Eq. (\ref{eq: first_order_e_spline_definition_1}) for $\tau \in[\tau_{b,k},\tau_{b,K}]$ with $\tau_{b,K}<\tau_{b,k}+\frac{L}{2}$, and from the hypothesis that $\varphi(\tau)$ reproduces the exponentials $e^{\pm j\omega_{m_0}\tau}$. Moreover, $(i)$ follows from the hypothesis that $0<\omega_{m_0}\leq\frac{\pi}{L}$ which is equivalent to $0<\frac{\omega_{m_0}L}{2}\leq\frac{\pi}{2}$, and from the assumption that $\tau_{b,K}-\tau_{b,k}<\frac{L}{2}$, which means that $0<\omega_{m_0}(\tau_{b,K}-\tau_{b,k})<\frac{\pi}{2}$, and hence $1-\cos(\omega_{m_0}(\tau_{b,K}-\tau_{b,1}))>1-\cos(\omega_{m_0}(\tau_{b,K}-\tau_{b,k}))$ $\forall k=2,...,K$.

In Fig.~\ref{fig:int_bursts_all_signals} we notice that in some cases (third burst of 2 Diracs), the spike $t_n$ may occur in the interval $(\tau_{b,1}, \tau_{b,K})$. Nevertheless, the condition in Eq. (\ref{eq:burst_integrator_inequality_1}) ensures that the sample at $t_{n+1}$ happens after $\tau_{b,K}$.

Similarly, since $f(\tau)>0$, Eq. (\ref{eq:burst_integrator_sample_cond_2}) is equivalent to:
\small
\begin{equation}
\vspace{-0.5em}
\label{eq:burst_integrator_inequality_3}
\int_{\tau_{b,1}}^{\tau_{b,1}+\frac{L}{2}} f(\tau)d\tau>\int_{\tau_{b,1}}^{t_{n+3}} f(\tau)d\tau,
\end{equation}
\normalsize
\vspace{-0.5em}
where the left-hand side can be expressed as:
\small
\begin{equation}
\label{eq:burst_integrator_inequality_4}
\begin{split}
&\int_{\tau_{b,1}}^{\tau_{b,1}+\frac{L}{2}} f(\tau)d\tau \aeq \sum_{k=1}^K \int_{\tau_{b,k}}^{\tau_{b,1}+\frac{L}{2}} x_k \varphi(\tau_{b,k}-\tau)d\tau \\
&\beq \frac{1}{\omega_{m_0}^2} \sum_{k=1}^K x_{b,k}[1-\cos(\omega_{m_0}(\frac{L}{2}-(\tau_{b,k}-\tau_{b,1})))]\\
&\cg \frac{1}{\omega_{m_0}^2} \sum_{k=1}^K x_{b,k}[1-\cos(\omega_{m_0}(\frac{L}{2}-(\tau_{b,K}-\tau_{b,1})))] \\
&\dg  \frac{K A_{\min}}{\omega_{m_0}^2}[1-\cos(\omega_{m_0}(\frac{L}{2}-(\tau_{b,K}-\tau_{b,1})))] \eg 5C_T,
\end{split}
\end{equation}
\normalsize
where $(a)$ follows from Eq. (\ref{eq:filtered input integrator}), $(b)$ follows from the definition of $\varphi(\tau_{b,k}-\tau)$ in Eq. (\ref{eq: first_order_e_spline_definition_1}) for $\tau \in (\tau_{b,k},\tau_{b,1}+\frac{L}{2})$, and $(c)$ follows from the hypothesis that $0<\omega_{m_0}\leq\frac{\pi}{L}$ which is equivalent to $0<\frac{\omega_{m_0}L}{2}\leq\frac{\pi}{2}$, and since $\tau_{b,k}-\tau_{b,1}<\frac{L}{2}$ $\forall k=2,...,K$. Moreover, $(d)$ holds since we assume $x_{b,1},...,x_{b,K}>0$, and $(e)$ follows from Eq. (\ref{eq:threshold_cond_summary_2}).

Finally, the right-hand side of Eq. (\ref{eq:burst_integrator_inequality_3}) is equivalent to:
\small
\begin{equation*}
\vspace{-0.5em}
\int_{\tau_{b,1}}^{t_{n+3}} f(\tau)d\tau  \feq 4C_T -  \int_{t_{n-1}}^{\tau_{b,1}} f(\tau)d\tau\ggl 5C_T \hl \int_{\tau_{b,1}}^{\tau_{b,1}+\frac{L}{2}} f(\tau)d\tau,
\end{equation*}
\normalsize
hence proving the result in Eq. (\ref{eq:burst_integrator_sample_cond_2}).

In these derivations, $(f)$ follows from Eq. (\ref{eq:non_uniform_samples integrator}), $(g)$ holds since $t_n>\tau_{b,1}>t_{n-1}$ and $(h)$ follows from Eq. (\ref{eq:burst_integrator_inequality_4}).

The conditions in Eq. (\ref{eq:burst_integrator_sample_cond_1}) and (\ref{eq:burst_integrator_sample_cond_2}) ensure that the output samples $y(t_{n+2})$ and $y(t_{n+3})$ have contributions only from all the Diracs in the $b^{th}$ burst. These samples can be computed using Eq. (\ref{eq:modified samples integrator summary}) and (\ref{eq:input burst Diracs}) for each channel $m$, as follows:
\small
\begin{equation}
\vspace{-0.5em}
\label{eq:sample_1_burst_same_sign}
\begin{split}
y(t_{n+2})&\aeq \sum_{k=1}^K x_{b,k}(\varphi*q_{\theta_{n+2}})(\tau_{b,k}-t_{n+1}).
\end{split}
\end{equation}
\normalsize

Similarly, we can write $y(t_{n+3})$ as:
\vspace{-0.8em}
\small
\begin{equation}
\label{eq:sample_2_burst_same_sign}
y(t_{n+3}) =\sum_{k=1}^K x_{b,k}(\varphi*q_{\theta_{n+3}})(\tau_{b,k}-t_{n+2}).
\vspace{-0.8em}
\end{equation}
\normalsize

For each channel $m$, the signal $(\varphi*q_{\theta_{n+2}})(t-t_{n+1})$ is a linear combination of the exponentials $e^{j\omega_{m_0} t}$ and $e^{j\omega_{m_1}t}$, for $t \in (t_{n+2}-\frac{L}{2},t_{n+1})$, given Eq. (\ref{eq: first_order_e_spline_definition_1}) and Eq. (\ref{eq:box function integrator}). Similarly, $(\varphi*q_{\theta_{n+3}})(t-t_{n+2})$  is a linear combination of the exponentials $e^{j\omega_{m_0} t}$ and $e^{j\omega_{m_1}t}$, for $t \in (t_{n+3}-\frac{L}{2},t_{n+1})$. Therefore, in the interval $(t_{n+3}-\frac{L}{2},t_{n+1})$, where there are no knots of either $(\varphi*q_{\theta_{n+2}})(t-t_{n+1})$ or $(\varphi*q_{\theta_{n+3}})(t-t_{n+2})$, we use the proof in Section \ref{subsubsection: Exponential reproducing kernels} to find unique $c_{m_i,2}$ and $c_{m_i,3}$ such that:
\small
\begin{equation}
\vspace{-0.5em}
\label{eq:recon exp burst same sign}
c_{m_i,2}(\varphi*q_{\theta_{n+2}})(t-t_{n+1})+ c_{m_i,3}(\varphi*q_{\theta_{n+3}})(t-t_{n+2})= e^{j\omega_{m_i}t}, 
\end{equation}
\normalsize
for $i \in \{0,1\}$, $t \in [t_{n+3}-\frac{L}{2}, t_{n+1}]$, $m_0=m$ and $m_1=2K-1-m$ (which ensures $\omega_{m_1}=-\omega_{m_0}$).

Then, for each channel $m$ we can compute the signal moments as before:
\small
\begin{equation*}
\vspace{-0.5em}
\begin{split}
&s_{m_i} = c_{m_i,2}y(t_{n+2})+c_{m_i,3}y(t_{n+3}) \\
&\aeq \sum_{k=1}^K x_{b,k} \sum_{l=2}^{3} c_{m_i,l}(\varphi*q_{\theta_{l}})(\tau_{b,k}-t_{l+n-1})\beq \sum_{k=1}^K x_{b,k} e^{j\omega_{m_i} \tau_{b,k}},
\end{split}
\end{equation*}
\normalsize
where $i \in \{0,1\}$, $m_0=m$ and $m_1=2K-1-m$.

In the derivations above, (a) follows from Eq. (\ref{eq:sample_1_burst_same_sign}) and (\ref{eq:sample_2_burst_same_sign}), and (b) from $\tau_{b,1},...,\tau_{b,K} \in(t_{n+3}-\frac{L}{2},t_{n+1})$ and the fact that Eq. (\ref{eq:recon exp burst same sign}) holds within this interval.
We can then uniquely retrieve the $2K$ input parameters of the $b^{th}$ burst from the $2K$ signal moments $s_{m_i}$ of all channels, using Prony's method.

Finally, we make the observation that the inequalities in Eq. (\ref{eq:threshold_cond_summary_1}) and Eq. (\ref{eq:threshold_cond_summary_2}) impose additional constraints on the maximum separation between the Diracs in a burst $b$, namely on $\tau_{b,K}-\tau_{b,1}$. Specifically, we need to impose:
\small
\begin{equation*}
5\int_{\tau_{b,1}}^{\tau_{b,K}} f(\tau)d\tau < \int_{\tau_{b,1}}^{\tau_{b,1}+\frac{L}{2}} f(\tau) d\tau,
\end{equation*}
\normalsize
which may give different constraints on the Dirac separation according to the filter characteristics, $A_{\max}$ and $A_{\min}$.

\vspace{-1em}


\begin{thebibliography}{10}

\bibitem{843002}
M.~Unser.
\newblock {Sampling-50 years after Shannon}.
\newblock {\em Proceedings of the IEEE}, 88(4):569--587, Apr 2000.

\bibitem{1697831}
C.~E. Shannon.
\newblock Communication in the presence of noise.
\newblock {\em Proceedings of the IRE}, 37(1):10--21, Jan 1949.

\bibitem{1580791}
E.~J. Cand{\`{e}}s, J.~Romberg, and T.~Tao.
\newblock {Robust uncertainty principles: exact signal reconstruction from
  highly incomplete frequency information}.
\newblock {\em IEEE Transactions on Information Theory}, 52(2):489--509, Feb
  2006.

\bibitem{1614066}
D.~L. Donoho.
\newblock {Compressed sensing}.
\newblock {\em IEEE Transactions on Information Theory}, 52(4):1289--1306, Apr
  2006.

\bibitem{candes-granda}
E.~J. Cand{\`{e}}s and C.~Fernandez{-}Granda.
\newblock Towards a mathematical theory of super-resolution.
\newblock {\em CoRR}, abs/1203.5871, 2012.

\bibitem{4156380}
P.~L. Dragotti, M.~Vetterli, and T.~Blu.
\newblock {Sampling Moments and Reconstructing Signals of Finite Rate of
  Innovation: Shannon Meets Strang-Fix}.
\newblock {\em IEEE Transactions on Signal Processing}, 55(5):1741--1757, May
  2007.

\bibitem{1003065}
M.~Vetterli, P.~Marziliano, and T.~Blu.
\newblock {Sampling signals with finite rate of innovation}.
\newblock {\em IEEE Transactions on Signal Processing}, 50(6):1417--1428, Jun
  2002.

\bibitem{5686950}
R.~Tur, Y.~C. Eldar, and Z.~Friedman.
\newblock {Innovation Rate Sampling of Pulse Streams With Application to
  Ultrasound Imaging}.
\newblock {\em IEEE Transactions on Signal Processing}, 59(4):1827--1842, Apr
  2011.

\bibitem{4483755}
Y.~M. Lu and M.~N. Do.
\newblock {A Theory for Sampling Signals from a Union of Subspaces}.
\newblock {\em IEEE Transactions on Signal Processing}, 56(6):2334--2345, Jun
  2008.

\bibitem{4682542}
C.~S. {Seelamantula} and M.~{Unser}.
\newblock A generalized sampling method for finite-rate-of-innovation-signal
  reconstruction.
\newblock {\em IEEE Signal Processing Letters}, 15:813--816, 2008.

\bibitem{8682626}
R.~{Alexandru} and P.~L. {Dragotti}.
\newblock Time-based sampling and reconstruction of non-bandlimited signals.
\newblock In {\em ICASSP 2019 - 2019 IEEE International Conference on
  Acoustics, Speech and Signal Processing (ICASSP)}, pages 7948--7952, May
  2019.

\bibitem{alexandrusampta}
R.~{Alexandru} and P.~L. {Dragotti}.
\newblock Time encoding and perfect recovery of non-bandlimited signals with an
  integrate-and-fire system.
\newblock In {\em SampTA 2019 - 13th International Conference on Sampling
  Theory and Applications}, Jul 2019.

\bibitem{6770840}
B.~F. {Logan}.
\newblock Information in the zero crossings of bandpass signals.
\newblock {\em The Bell System Technical Journal}, 56(4):487--510, Apr 1977.

\bibitem{soton252088}
R~Steele.
\newblock {\em Delta Modulation Systems}.
\newblock Pentech Press \& Halsted Press, 1975.

\bibitem{1344228}
A.~A. {Lazar} and L.~T. {Toth}.
\newblock Perfect recovery and sensitivity analysis of time encoded bandlimited
  signals.
\newblock {\em IEEE Transactions on Circuits and Systems I: Regular Papers},
  51(10):2060--2073, Oct 2004.

\bibitem{adrian1928basis}
E.D.A. Adrian.
\newblock {\em The basis of sensation: the action of the sense organs}.
\newblock Hafner, 1928.

\bibitem{Dayan:2005:TNC:1205781}
P.~Dayan and L.~F. Abbott.
\newblock {\em Theoretical Neuroscience: Computational and Mathematical
  Modeling of Neural Systems}.
\newblock The MIT Press, 2005.

\bibitem{Gerstner:2002:SNM:583784}
W.~Gerstner and W.~Kistler.
\newblock {\em Spiking Neuron Models: An Introduction}.
\newblock Cambridge University Press, New York, NY, USA, 2002.

\bibitem{5537149}
T.~{Delbr\"{u}ck}, B.~{Linares-Barranco}, E.~{Culurciello}, and C.~{Posch}.
\newblock Activity-driven, event-based vision sensors.
\newblock In {\em Proceedings of 2010 IEEE International Symposium on Circuits
  and Systems}, pages 2426--2429, May 2010.

\bibitem{Lazar05timeencoding}
A.~A. Lazar.
\newblock {Time encoding with an integrate-and-fire neuron with a refractory
  period}.
\newblock {\em Neurocomputing}, 65:65--66, 2005.

\bibitem{1201780}
A.~A. Lazar and L.~T. Toth.
\newblock {Time encoding and perfect recovery of bandlimited signals}.
\newblock In {\em 2003 IEEE International Conference on Acoustics, Speech, and
  Signal Processing, 2003. Proceedings. (ICASSP '03).}, volume~6, pages
  VI--709, Apr 2003.

\bibitem{5709990}
A.~A. Lazar and E.~A. Pnevmatikakis.
\newblock {Video Time Encoding Machines}.
\newblock {\em IEEE Transactions on Neural Networks}, 22(3):461--473, Mar 2011.

\bibitem{Feichtinger2012}
H.~Feichtinger, J.~Pr\'{\i}ncipe, J.~Romero, A.~Singh~Alvarado, and G.~Velasco.
\newblock {Approximate Reconstruction of Bandlimited Functions for the
  Integrate and Fire Sampler}.
\newblock {\em Adv. Comput. Math.}, 36(1):67--78, Jan 2012.

\bibitem{1415989}
A.~A. {Lazar}, E.~K. {Simonyi}, and L.~T. {Toth}.
\newblock Fast recovery algorithms for time encoded bandlimited signals.
\newblock In {\em Proceedings. (ICASSP '05). IEEE International Conference on
  Acoustics, Speech, and Signal Processing, 2005.}, volume~4, pages
  iv/237--iv/240 Vol. 4, Mar 2005.

\bibitem{Adam19}
K.~{Adam}, A.~{Scholefield}, and M.~{Vetterli}.
\newblock {Multi-channel time encoding for improved reconstruction of
  bandlimited signals}.
\newblock In {\em 2019 IEEE International Conference on Acoustics, Speech, and
  Signal Processing. Proceedings. (ICASSP '19).}, May 2019.

\bibitem{FlorescuC15}
D.~Florescu and D.~Coca.
\newblock A novel reconstruction framework for time-encoded signals with
  integrate-and-fire neurons.
\newblock {\em Neural Computation}, 27(9):1872--1898, 2015.

\bibitem{GONTIER201463}
D.~Gontier and M.~Vetterli.
\newblock {Sampling based on timing: Time encoding machines on shift-invariant
  subspaces}.
\newblock {\em Applied and Computational Harmonic Analysis}, 36(1):63 -- 78,
  2014.

\bibitem{AldroubiGrochenig}
A.~Aldroubi and K.~Gr\"{o}chenig.
\newblock {Nonuniform Sampling and Reconstruction in Shift-Invariant Spaces}.
\newblock {\em SIAM Review}, 43(4):585--620, 2001.

\bibitem{article}
H.~Feichtinger and K.~Gr{\"o}chenig.
\newblock Theory and practice of irregular sampling.
\newblock {\em Wavelets: Mathematics and Applications}, Jan 1994.

\bibitem{330352}
M.~Unser and A.~Aldroubi.
\newblock {A general sampling theory for nonideal acquisition devices}.
\newblock {\em IEEE Transactions on Signal Processing}, 42(11):2915--2925, Nov
  1994.

\bibitem{LAP09b}
A.~A. Lazar and E.~A. Pnevmatikakis.
\newblock Reconstruction of sensory stimuli encoded with integrate-and-fire
  neurons with random thresholds.
\newblock {\em EURASIP Journal on Advances in Signal Processing}, 2009.

\bibitem{1329542}
I.~{Maravic}, M.~{Vetterli}, and K.~{Ramchandran}.
\newblock {Channel estimation and synchronization with sub-Nyquist sampling and
  application to ultra-wideband systems}.
\newblock In {\em 2004 IEEE International Symposium on Circuits and Systems
  (ISCAS)}, volume~5, pages V--V, May 2004.

\bibitem{7857059}
G.~{Baechler}, A.~{Scholefield}, L.~{Baboulaz}, and M.~{Vetterli}.
\newblock Sampling and exact reconstruction of pulses with variable width.
\newblock {\em IEEE Transactions on Signal Processing}, 65(10):2629--2644, May
  2017.

\bibitem{7736135}
H.~{Pan}, T.~{Blu}, and M.~{Vetterli}.
\newblock {Towards Generalized FRI Sampling With an Application to Source
  Resolution in Radioastronomy}.
\newblock {\em IEEE Transactions on Signal Processing}, 65(4):821--835, Feb
  2017.

\bibitem{7465789}
X.~{Wei} and P.~L. {Dragotti}.
\newblock {FRESH--FRI-Based Single-Image Super-Resolution Algorithm}.
\newblock {\em IEEE Transactions on Image Processing}, 25(8):3723--3735, Aug
  2016.

\bibitem{Onativia2013}
J.~O{\~{n}}ativia, S.~R. Schultz, and P.~L. Dragotti.
\newblock A finite rate of innovation algorithm for fast and accurate spike
  detection from two-photon calcium imaging.
\newblock {\em Journal of Neural Engineering}, 10(4):046017, jul 2013.

\bibitem{1408194}
M.~{Unser}.
\newblock {Cardinal exponential splines: part II - think analog, act digital}.
\newblock {\em IEEE Transactions on Signal Processing}, 53(4):1439--1449, April
  2005.

\bibitem{Prony}
R.~Prony.
\newblock {Essai exp\'{e}rimental et analytique sur les lois de la
  dilatabilit\'{e} de fluides \'{e}lastiques et sur celles de la force
  expansive de la vapeur de l'eau et de la vapeur de l'alkool, \`{a}
  diff\'{e}rentes temp\'{e}ratures}.
\newblock {\em Journal de l'{\'E}cole Polytechnique}, 1(22):24--76.

\bibitem{coderoxana}
{R. Alexandru}.
\newblock {Code for Reconstructing Classes of Non-bandlimited Signals from Time
  Encoded Information}.
\newblock Dec 2019.
\newblock Available:
  \url{https://github.com/rialexandru01/Reconstructing-Classes-of-Non-bandlimited-Signals-from-Time-Encoded-Information}.

\bibitem{838174}
Z.~{Cvetkovic} and I.~{Daubechies}.
\newblock Single-bit oversampled a/d conversion with exponential accuracy in
  the bit-rate.
\newblock In {\em Proceedings DCC 2000. Data Compression Conference}, pages
  343--352, March 2000.

\bibitem{Lazar2008}
A.~A. Lazar and E.~A. Pnevmatikakis.
\newblock Faithful representation of stimuli with a population of
  integrate-and-fire neurons.
\newblock {\em Neural Computation}, 20(11):2715--2744, 2008.
\newblock PMID: 18533815.

\bibitem{799930}
M.~Unser.
\newblock {Splines: a perfect fit for signal and image processing}.
\newblock {\em IEEE Signal Processing Magazine}, 16(6):22--38, Nov 1999.

\bibitem{1408193}
M.~Unser and T.~Blu.
\newblock {Cardinal exponential splines: part I - theory and filtering
  algorithms}.
\newblock {\em IEEE Transactions on Signal Processing}, 53(4):1425--1438, Apr
  2005.

\bibitem{Strang2011}
G.~Strang and G.~Fix.
\newblock {\em {A Fourier Analysis of the Finite Element Variational Method}},
  pages 793--840.
\newblock Springer Berlin Heidelberg, Berlin, Heidelberg, 2011.

\bibitem{Bolzano}
S.~B. Russ.
\newblock {A translation of Bolzano's paper on the intermediate value theorem}.
\newblock {\em Historia Mathematica - HIST MATH}, 7:156--185, May 1980.

\bibitem{Stoica104835}
P.~Stoica and R.~Moses.
\newblock {\em Spectral Analysis of Signals}.
\newblock 2005.

\bibitem{6578165}
J.~A. Urig\"{u}en, T.~Blu, and P.~L. Dragotti.
\newblock {FRI Sampling With Arbitrary Kernels}.
\newblock {\em IEEE Transactions on Signal Processing}, 61(21):5310--5323, Nov
  2013.

\bibitem{4472241}
T.~{Blu}, P.~{Dragotti}, M.~{Vetterli}, P.~{Marziliano}, and L.~{Coulot}.
\newblock Sparse sampling of signal innovations.
\newblock {\em IEEE Signal Processing Magazine}, 25(2):31--40, Mar 2008.

\bibitem{1488}
J.~A. {Cadzow}.
\newblock Signal enhancement-a composite property mapping algorithm.
\newblock {\em IEEE Transactions on Acoustics, Speech, and Signal Processing},
  36(1):49--62, Jan 1988.

\end{thebibliography}

\end{document}